\documentclass[11pt]{amsart}
\usepackage[english]{babel}

\usepackage{physics}
\usepackage{braket}
\usepackage{dsfont}
\usepackage{graphicx}
\usepackage{amssymb}
\usepackage{tikz}
\usetikzlibrary{matrix}
\usetikzlibrary{shapes,arrows,chains, decorations.pathmorphing, decorations.pathreplacing, arrows.meta}
\tikzset{quantum/.style={decorate, decoration=snake}}

\usepackage[margin=1in]{geometry}
\usepackage[pagebackref, colorlinks = true, linkcolor = blue, urlcolor  = blue, citecolor = red]{hyperref}

\renewcommand{\epsilon}{\varepsilon}
\renewcommand{\phi}{\varphi}

\newtheorem{thm}{Theorem}[section]
\newtheorem*{thm*}{Theorem}
\newtheorem{lem}[thm]{Lemma}
\newtheorem{cor}[thm]{Corollary}

\newtheorem{defi}[thm]{Definition}
\newtheorem{prop}[thm]{Proposition}
\newtheorem{remark}[thm]{Remark}

\newcommand{\PVBB}{$PV^f_{\mathrm{meas}}$~}
\newcommand{\PVroute}{$\widetilde{PV}^f_{\mathrm{route}}$~}
\newcommand{\PVroutemod}{${PV}^f_{\mathrm{route}}$~}

\begin{document}
\title[Single-qubit position verification protocol]{A single-qubit position verification protocol that is secure against multi-qubit attacks$^1$}

\author{Andreas Bluhm}
\email{bluhm@math.ku.dk}
\address{QMATH, Department of Mathematical Sciences, University of Copenhagen, Universitetsparken 5, 2100 Copenhagen, Denmark}

\author{Matthias Christandl}
\email{christandl@math.ku.dk}
\address{QMATH, Department of Mathematical Sciences, University of Copenhagen, Universitetsparken 5, 2100 Copenhagen, Denmark}

\author{Florian Speelman}
\email{f.speelman@uva.nl}
\address{QuSoft \& Informatics Institute, University of Amsterdam, Science Park 904, Amsterdam, the Netherlands}

\date{\today}

\maketitle

\textbf{The position of a device or agent is an important security credential in today’s society, both online and in the real world. Unless in direct proximity, however, the secure verification of a position is impossible without further assumptions. This is true classically \cite{Chandran2009A}, but also in any future quantum-equipped communications infrastructure \cite{buhrman2014positionA}. We show in this work that minimal quantum resources, in the form of a single qubit, combined with classical communication are sufficient to thwart quantum adversaries that pretend to be at a specific position and have the ability to coordinate their action with entanglement. More precisely, we show  that the adversaries using an increasing amount of entanglement can be combatted solely by increasing the number of classical bits used in the protocol. The presented protocols are noise-robust and within reach of current quantum technology. }

\footnotetext[1]{This version of the article has been accepted for publication, after peer review and is subject to Springer Nature’s AM terms of use, but is not the Version of Record and does not reflect post-acceptance improvements, or any corrections. The Version of Record is available online at: \url{https://doi.org/10.1038/s41567-022-01577-0}}

The difficulty in achieving the verification of a position is best appreciated by considering certain secure-looking protocols and then understanding how they can be broken. For simplicity of the presentation we will consider the verification of the position of an untrusted agent being at midpoint between two verifiers. A protocol for position verification consists of the verifiers each sending messages to the agent who is asked to send responses back. The verification is successful if the responses satisfy certain conditions and if the timing of the signals is right (say in accordance with the speed of light) (see Figure \ref{fig:line}).

\begin{figure}[t]
    \centering
    \begin{tikzpicture}[node distance=6cm, auto]
        \path [draw] (0,-0.5)  -- (0, 0.5) coordinate[label={above:$V_0$}](A) -- (0,0) -- (4, 0) -- (4,0.5)coordinate[label={above:$P$}](B) -- (4,-0.5)coordinate[label={below:$z$}](C) -- (4,0) -- (8,0) -- (8,-0.5) -- (8,0.5)coordinate[label={above:$V_1$}](D);
        \path [draw, dashed, red] (2,0.5) coordinate[label={above:Alice}](E) -- (2,-0.5);
        \path [draw, dashed, red] (6,0.5) coordinate[label={above:Bob}](F) -- (6,-0.5);
        
        \node [below=2cm of A] (G) {$V_0$};
        \node [below=1cm of C] (H) {\phantom{X}};
        \node [below=2cm of D] (I) {$V_1$};
        \node [below=2cm of E] (J) {\phantom{X}};
         \node [below=2cm of F] (K) {\phantom{X}};
         \node [below=3.5cm of H] (L) {};
         \node [below=1.5cm of J, red] (M) {Alice};
          \node [below=1.5cm of K, red] (N) {Bob};
          \node [below=3.5cm of M, red] (O) {Alice};
          \node [below=3.5cm of N, red] (P) {Bob};
          \node [below=7.5cm of G] (Q) {$V_0$};
          \node [below=7.5cm of I] (R) {$V_1$};
          \node [left=3cm of G] (S) {};
          \node [left=3cm of Q] (T) {};
          \node [below=1cm of Q] (U) {};
          \node [below=1cm of R] (V) {};
          
          \draw [-{Triangle[scale=1.5]}] (G) -- (M) node[midway] (x) {$x \in \{0,1\}^n$};
          \draw [-{Triangle[scale=1.5]}] (M) -- (P) node[near start, below] (x1) {$x$\phantom{aaa}};
          \draw [-{Triangle[scale=1.5]}] (P) -- (R) node[midway] (x2) {$f(x,y)$} ;
          \draw [-{Triangle[scale=1.5]}] (I) -- (N) node[midway, above] (y) {\phantom{aaaaaaaaaaaaaaaaa}$y \in \{0,1\}^n$};
          \draw [-{Triangle[scale=1.5]}] (N) -- (O) node[near start, below] (y1) {\phantom{aaa}$y$}; 
          \draw [-{Triangle[scale=1.5]}] (O) -- (Q) node[midway, above] (y2) {$f(x,y)$\phantom{aaaaaa}} ;
          \draw [-{Triangle[scale=1.5]}] (M) -- (O) node[midway, left] (x3) {$x$};
          \draw [-{Triangle[scale=1.5]}] (N) -- (P) node[midway, right] (y3) {$y$};
          \draw [-{Triangle[scale=1.5]}] (S) -- (T) node[midway, right] (t) {time};
          \draw [-{Triangle[scale=1.5]}] (U) -- (V) node[midway, below] (t) {position};
    \end{tikzpicture}
    \caption{\textbf{Setup for PV in one spatial dimension (above) and classical attack (below).} The honest agent $P$ is at position $z$, whereas the verifiers $V_0$, $V_1$ are to her left and right. For an attack, $P$ is replaced by the attackers Alice and Bob, which are not at $z$, but in between $z$ and the verifiers. Upon receiving $x \in \{0,1\}^n$, Alice copies the string and sends a copy on to Bob while Bob does the same with $y \in \{0,1\}^n$. Both attackers compute the function $f$ and send the result back to the closest verifier. From the verifiers' point of view, they are indistinguishable from an honest prover $P$ at $z$.}
    \label{fig:line}
\end{figure}
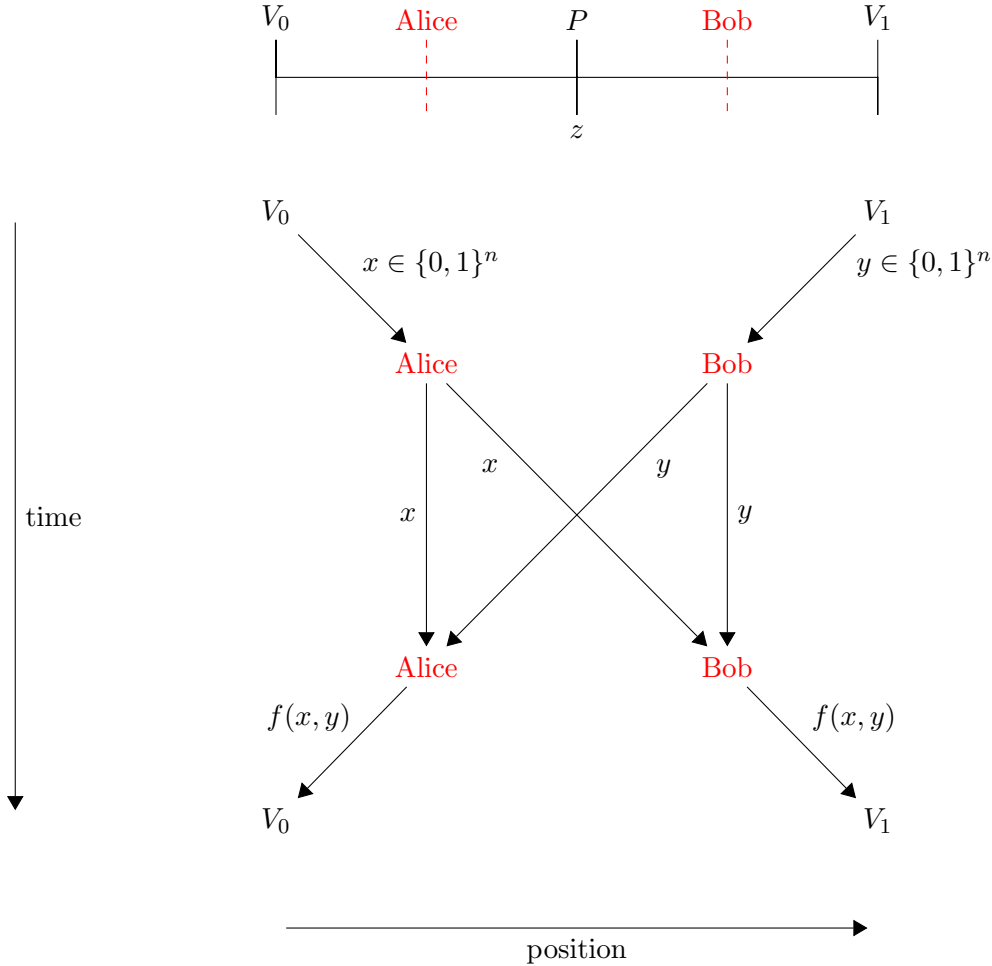

A first attempt for a secure protocol could consist of a Boolean function $f$ taking the message $x$ from verifier 0 and $y$ (both  $n$-bit strings) from verifier 1 as input and sending the bit $f(x, y)$ back to the verifiers. In order for the agent to return the correct answer, clearly (for most functions $f$) both $x$ and $y$ are needed, but if the agent was not at midpoint but, say, closer to verifier 0, the agent could never both receive $y$ and send the answer back to verifier 1 in time. Indeed, breaking the protocols entails not one attacking agent, but two, one of which is closer to verifier 0 and one which is closer to verifier 1. Customarily called Alice and Bob, the attacking agents both intercept the input from the verifier they are closest to. Each keeps a copy and forwards another copy of the input to the other partner in crime. When they hold both inputs in hand, they compute the function and return the function value just in time to their respective verifiers (see Figure \ref{fig:line}). 

This simple attack is indeed the basis of why position verification is not possible in the classical world. Note, however, that the attack directly uses the copying of information. This opens up the possibility of devising protocols based on the exchange of quantum information instead, whose copying is more restricted due to the no-cloning theorem \cite{Kent2011A, Malaney2010A}. As Alice and Bob can agree on an attack strategy prior to the start of the protocol, however, they can also distribute entangled particles in order to later coordinate their action. Still, has the balance now tipped and position verification become possible? The plain answer is no \cite{buhrman2014positionA}, as Alice and Bob can immediately upon receipt of the quantum particles engage in an elaborate scheme of back and forth teleportation (with only a single round of crossing classical communication), known as instantaneous non-local computation \cite{Vaidman2003A}. In a sense, this means `game over' for position verification --- if only the back and forth teleportation was not so expensive (doubly exponentially many EPR pairs in terms of the size of the protocol). This bound was brought down to singly exponential by use of port-based teleportation instead of standard quantum teleportation \cite{Beigi2011A}. Note that carrying out such attacks is still prohibitively expensive for the attackers. Therefore, such attacks could be seen as unrealistic, forcing us to ask again whether position verification is after all viable in the quantum world.

We give a partial answer to this question by showing that there are protocols that enhance the above classical protocol by a single qubit and that withstand attacks involving roughly $n$ qubits. A specific efficient protocol can withstand attacks of $\log n$ qubits. Thereby, we obtain security of position verification where the ratio of quantum resources required for the attack and of the honest agent is unbounded. 

This is a qualitative improvement over prior work, which did show some level of security for an $n$-qubit protocol inspired by the BB84 protocol \cite{buhrman2014positionA, Beigi2011A, tomamichel2013monogamyA, Ribeiro2015A}, but for which an attacker only needs one EPR pair per qubit involved.
Other previous works proposed new protocols \cite{Lau2011A, Chakraborty2015A, Malaney2016A, Das2017A, Gonzales2019A}, sometimes with the same scaling as BB84 protocol, and sometimes without security proofs showing how efficient the attacks can be; explicit efficient attacks can be constructed for several of these \cite{Buhrman2013A, Speelman2016A, Gonzales2019A, Olivo2020A}.
Complementary to our results are works that study other security models \cite{Kent2010A, Gao2013A, Unruh2014A}, and that introduce techniques to increase the robustness to photon loss \cite{Qi2015A, Lim2016A, Allerstorfer2021A}.
For the security analysis of a different protocol, see the recent independent work by Junge et al.\ \cite{Junge2021A}; we compare our results in Section \ref{sec:attack-model} of the Supplementary Information. The routing protocol was introduced by Kent et al.\ \cite{Kent2011A} and studied further by Buhrman et al.\ \cite{Buhrman2013A}. We build on the proof strategy of the latter work.

We consider in this work two closely related protocols, which we will dub the `routing' and `measuring' protocols, which are direct enhancements of the classical protocols explained. The protocols are thus specified by a Boolean function $f$ on $2n$ bits. In addition to the verifiers choosing random inputs $x$ and $y$ respectively, in both protocols verifier 0 will prepare a qubit chosen randomly from one of the BB84 states: $\ket{0}, \ket{1}, \ket{+}, \ket{-}$ and send it to the agent along with $x$. This could for instance be a single polarized photon sent in free space. In the routing protocol, the agent is asked to return the qubit unchanged to the verifier with number $f(x, y)$. Concretely, if $f(x, y)=1$, the verifier could let the photon pass to verifier 1 and if $f(x, y)=0$ use a mirror to reflect the photon back to verifier 0. The verifier could then measure the qubit and check whether the measurement result is consistent with the preparation. The protocol is illustrated in Figure \ref{fig:protocol-route}. In the `measuring' protocol, instead of routing the qubit, the agent is asked to measure in the $\ket{0}, \ket{1}$ basis in case $f(x, y)=0$ and in the $\ket{+}, \ket{-}$ basis if $f(x,y)=1$ and to return the measurement result to both verifiers. The protocol is illustrated in Figure \ref{fig:protocol-measure}.

\begin{figure}[b]
    \centering
    \begin{tikzpicture}[node distance=3cm, auto]
    \node (A) {$V_0$};
    \node [left=1cm of A] {};
    \node [right=of A] (B) {};
    \node [right=of B] (C) {$V_1$};
    \node [right=1cm of C] {};
    \node [below=of A] (D) {};
    \node [below=of B] (E) {$P$};
    \node [below=of C] (F) {};
    \node [below=of D] (G) {$V_0$};
    \node [below=of E] (H) {};
    \node [below=of F] (I) {$V_1$};
    \node [left= 6cm of E] (J) {};
    \node [below= 3cm of J] (K) {};
    \node [above= 3cm of J] (L) {};
    \node [below= 1cm of G] (X) {};
    \node [below= 1cm of I] (Y) {};

    \draw [-{Triangle[scale=1.5]}, transform canvas={xshift=- 2pt, yshift = -2 pt}, quantum] (A) -- (E) ;
    \draw [-{Triangle[scale=1.5]}, transform canvas={xshift=+ 2pt, yshift = +2 pt}] (A) -- (E) node[midway] (x) {$x \in \{0,1\}^n$};
    \draw [-{Triangle[scale=1.5]}] (C) -- (E);
    \draw [-{Triangle[scale=1.5]}, quantum] (E) -- (I) node[midway] (q) {$Q$ if $f(x,y) = 1$};
    \draw [-{Triangle[scale=1.5]},, quantum] (E) -- (G);

    \draw [-{Triangle[scale=1.5]}] (L) -- (K) node[midway] {time};
    \draw [-{Triangle[scale=1.5]}] (X) -- (Y) node[midway] {position};

    \node[left=1cm of x, transform canvas={xshift=+ 2pt, yshift = +2 pt}] {qubit $Q$};
    \node[right = 2.5cm of x, transform canvas={xshift=+ 2pt, yshift = +2 pt}] {$y \in \{0,1\}^n$};
    \node[left = 3.3cm of q] {$Q$ if $f(x,y) = 0$};
\end{tikzpicture}
\caption{\textbf{The routing protocol.} In the protocol, the verifier $V_0$ prepares a qubit $Q$ in one of the four BB84 states uniformly at random. Subsequently, $V_0$ sends $Q$ together with a random $n$-bit string $x$  to the agent $P$ at position $z$ and~$V_1$ sends a random $n$-bit string $y$. All communication happens at the speed of light and the timing is such that $Q$, $x$ and $y$ reach position $z$ at the same time.
Depending on the outcome $f(x,y)$ of a previously agreed upon Boolean function $f$ on $2n$ bits, the prover has to send the qubit $Q$ received immediately to either verifier $V_0$ or $V_1$. The qubit $Q$ has to reach the verifiers on time, i.e.\ the time of arrival at $V_{f(x,y)}$ has to be consistent with $Q$ being sent from $z$ at the speed of light right after $Q$ has reached $z$. Straight lines correspond to classical information, while undulated lines correspond to quantum information being sent.}
\label{fig:protocol-route}
\end{figure}
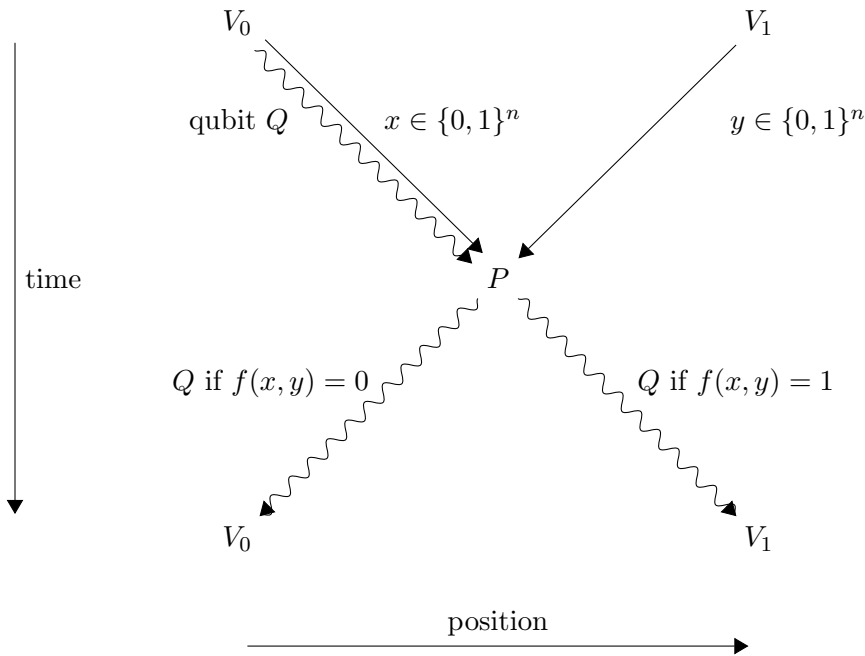

\begin{figure}[b]
    \centering
    \begin{tikzpicture}[node distance=3cm, auto]
    \node (A) {$V_0$};
    \node [left=1cm of A] {};
    \node [right=of A] (B) {};
    \node [right=of B] (C) {$V_1$};
    \node [right=1cm of C] {};
    \node [below=of A] (D) {};
    \node [below=of B] (E) {$P$};
    \node [below=of C] (F) {};
    \node [below=of D] (G) {$V_0$};
    \node [below=of E] (H) {};
    \node [below=of F] (I) {$V_1$};
    \node [left= 6cm of E] (J) {};
    \node [below= 3cm of J] (K) {};
    \node [above= 3cm of J] (L) {};
    \node [below= 1cm of G] (X) {};
    \node [below= 1cm of I] (Y) {};

    \draw [-{Triangle[scale=1.5]}, transform canvas={xshift=- 2pt, yshift = -2 pt}, quantum] (A) -- (E) ;
    \draw [-{Triangle[scale=1.5]}, transform canvas={xshift=+ 2pt, yshift = +2 pt}] (A) -- (E) node[midway] (x) {$x \in \{0,1\}^n$};
    \draw [-{Triangle[scale=1.5]}] (C) -- (E);
    \draw [-{Triangle[scale=1.5]}] (E) -- (I) node[midway] (q) {$b \in \{0,1\}$};
    \draw [-{Triangle[scale=1.5]}] (E) -- (G);

    \draw [-{Triangle[scale=1.5]}] (L) -- (K) node[midway] {time};

    \node[left=1cm of x, transform canvas={xshift=+ 2pt, yshift = +2 pt}] {qubit $Q$};
    \node[right = 2.5cm of x, transform canvas={xshift=+ 2pt, yshift = +2 pt}] {$y \in \{0,1\}^n$};
    \node[left = 3.3cm of q] {$b \in \{0,1\}$};
    \draw [-{Triangle[scale=1.5]}] (X) -- (Y) node[midway] {position};
\end{tikzpicture}
\caption{\textbf{The measuring protocol.} In the protocol, the verifiers $V_0$ and~$V_1$ choose two random bit strings $x$, $y$ of length~$n$. If $f(x,y)=0$, $V_0$ prepares a qubit $Q$ in one of the computational basis states with equal probability, otherwise, $V_0$ prepares $Q$ in one of the Hadamard basis states. Then, $V_0$ sends $Q$ and $x$ to $P$, $V_1$ sends $y$, and the timing is such that $Q$, $x$ and $y$ reach position $z$ at the same time.
If $f(x,y) = 0$, the prover measures $Q$ in the computational basis, otherwise in the Hadamard basis. The outcome bit $b$ of the measurement is subsequently sent back to both verifiers. It has to reach the verifiers on time, i.e.\ the time of arrival of $b$ has to be consistent with $b$ being sent from $z$ at the speed of light right after $Q$ has reached $z$. Straight lines correspond to classical information, while undulated lines correspond to quantum information being sent.}
\label{fig:protocol-measure}
\end{figure}
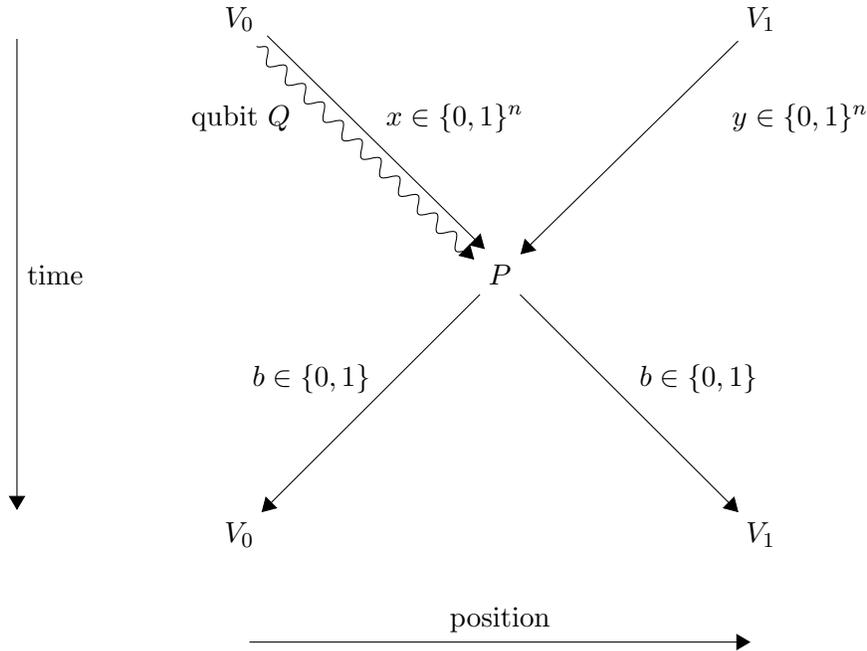

In a sense, the only difference between the protocols is who carries out the measurement. It turns out that our security arguments therefore only differ in a single place (for more information see the methods section). As is familiar from the security analysis of quantum key distribution protocols, the security analysis of the described prepare and measure protocols is equivalent to their natural entanglement-based versions, which is preferred in formal arguments due to their conceptual simplicity. Here, verifier 0 prepares an EPR pair and sends half of it to the agent and holds on to the other half as a reference qubit. The protocol is otherwise unchanged and in order for the verifier later to compare results, the verifier will measure the reference qubit.

Let us point out that the implementation of the protocols merely requires the honest parties to be able to prepare and measure BB84 states, a task that is routinely carried out in the context of quantum key distribution both in laboratories and commercially. Indeed, the least quantum-technological requirements are demanded from the agent or the agent's device in the routing protocol: namely to either to reflect a photon with a mirror or to measure it.

The routing protocol is even simpler than the measuring protocol in the sense that the honest agent needs to perform no measurements. On the other hand, the reply of the agent in the measuring protocol is completely classical, and here our security proof also applies to the setting where quantum information travels slowly, meaning that only classical messages travel at the speed of light. This requirement fits current technology better, where qubits are transmitted using fiber optics. Thus, both protocols have their pros and cons and it depends on the desired application to determine which one is better suited.

We can show that for an appropriate function $f$, both the routing and measuring protocol are secure if the attackers do not hold more than $n/2 - 5$ qubits each at the beginning of their attack, when strings $x$ and $y$ of length $n$ are sent by the verifiers. The most general form of the attacks is depicted in Figure \ref{fig:strategy-both}.  Moreover, the protocols can be repeated sequentially to make the probability that the attackers go unnoticed exponentially small. While we cannot give a concrete function $f$, we show that a uniformly random Boolean function will work with overwhelming probability. 

\begin{figure}
\centering
\begin{tikzpicture}[node distance=2cm, auto]
\node[red] (B) {Alice};
\node[right=7 cm of B, red] (C) {Bob};
\node[right=3.5 cm of B] (F) {};
\node[below= 1.5 cm of B] (J) {$A$~~$\tilde A$~~$A_c$};
\node[above=1cm of J] (J') {};
\node[right=1 cm of J] (J'') {\phantom{$\tilde A$}$x$};
\node[below= 1.5 cm of C] (M) {$B_c$~~$\tilde B$~~$B$};
\node[above=1cm of M] (M') {};
\node[left=1 cm of M] (M'') {$y$\phantom{$\tilde A$}};
\node[below= of J] (P) {$A$~~$\tilde A$~~$B_c$};
\node[below=1 cm of P] (P') {};
\node[right= 0.8 cm of P] (P'') {\phantom{$\tilde A$}$x, y$};
\node[below= of M] (S) {$A_c$~~$\tilde B$~~$B$};
\node[below=1 cm of S] (S') {};
\node[left =0.8 cm of S] (S'') {$x, y$\phantom{$\tilde A$}};
\node[left= 2 cm of J] (V) {};
\node[left= 2 cm of P] (W) {};
\node[above= 1.5 cm of V] (X) {};
\node[below= 1 cm of W] (Y) {};

\draw[-{Implies}, double, double distance= 1mm] (J') -- (J) node[midway, left] (u) {$U^x$\phantom{aa}};
\draw[-{Implies}, double, double distance= 1mm] (M') -- (M) node[midway, right] (v) {\phantom{aa}$V^y$};
\draw[-{Implies}, double, double distance= 1mm] (S) -- (S') node[midway, right] (l) {\phantom{aa}\textcolor{cyan}{$L^{xy}$}} node[midway, left] (sigma) {\textcolor{brown}{$\{\Sigma^{xy}, I - \Sigma^{xy}\}$}\phantom{aa}};
\draw[-{Implies}, double, double distance= 1mm] (P) -- (P') node[midway, left] (k) {\textcolor{cyan}{$K^{xy}$}\phantom{aa}} node[midway, right] (pi) {\phantom{aa}\textcolor{brown}{$\{\Pi^{xy}, I - \Pi^{xy}\}$}};
\draw[-{Triangle[scale=1.5]}, quantum] (J) -- (S);
\draw[-{Triangle[scale=1.5]}, quantum] (M) -- (P);
\draw[-{Triangle[scale=1.5]}] (J'') -- (S'');
\draw[-{Triangle[scale=1.5]}] (M'') -- (P'');
\draw[-{Triangle[scale=1.5]}] (X) -- (Y) node[midway, left] (t) {time};
\end{tikzpicture}
 \caption{\textbf{Attack strategies for the routing and the measuring protocol.} The parts in black are the same for both protocols, while cyan belongs to the routing protocol and brown to the measuring protocol. Straight lines correspond to classical information, while undulated lines correspond to quantum information being sent. We assume that Alice and Bob each have a qubit system $A$ and $B$, respectively. Moreover, Alice and Bob have local quantum registers $\tilde A$ and $\tilde B$. Due to the constraints imposed by special relativity, Alice and Bob are allowed one round of quantum communication, during which they can exchange systems $A_c$ and $B_c$. We assume that both Alice and Bob have the same number of qubits. At the beginning of the protocol, Alice intercepts $x$ and stores the qubit $Q$ in $A$, while Bob intercepts $y$. The most general attacks are as follows:
    $(1)$ Alice applies $U^x$ on $A\tilde A A_c$, Bob applies $V^y$ to $ B \tilde B B_c$.
    $(2)$ Alice sends $A_c$ and $x$ to Bob, Bob sends $B_c$ and $y$ to Alice.
    \textbf{Routing protocol:}
    $(3)$ Alice applies $K^{xy}$ on $A \tilde A B_c$ and Bob applies $L^{xy}$ on $B \tilde B A_c$.
    $(4)$ If $f(x,y) = 0$, Alice returns $A$ to verifier $V_0$, if $f(x,y) = 1$, Bob returns $B$ to verifier $V_1$. 
    \textbf{Measuring protocol:}
    $(3)$ Alice measures $\{\Pi^{xy}, I - \Pi^{xy}\}$ on $A \tilde A B_c$ and Bob measures $\{\Sigma^{xy}, I - \Sigma^{xy}\}$ on $B \tilde B A_c$.
    $(4)$ Alice sends her measurement outcome to verifier $0$, Bob sends his measurement outcome to verifier $1$.
    Here, we take all operators to be unitaries and the superscript indicates which classical strings the unitaries might depend on. }
\label{fig:strategy-both}
\end{figure}

Moreover, we consider the effect of noise on both protocols. The noise in the setup causes the honest agent not to succeed  with certainty, but to fail with a $1\%$ chance. In order to deal with the noise, the verifiers will repeat the protocol sequentially a number of times and accept if the protocol succeeds more than a fixed number of times. We can show that such protocols are still secure: An honest prover is rejected with a probability exponentially small in the number of repetitions. On the other hand, attackers controlling at most $n/2-5$ qubits at each round will succeed with probability exponentially small in the number of repetitions. This noise robustness of the single qubit protocols makes them interesting for near-term experimental implementations: for any reasonable bound on the number of qubits, i.e., a standard quantum bounded-storage assumption in cryptography, we have a secure protocol transmitting only a little more \emph{classical} information as well as a single qubit over a \emph{noisy} communication line.

Finally, we give lower bounds for concrete functions $f$, based on their communication complexity. For example, routing and measuring protocols using the inner product function are secure against attackers with at most $\log(n)/2 - 5$ qubits each. While these bounds for concrete functions are exponentially worse than for random functions, they still exhibit the feature that the ratio of the quantum resources the attackers need compared to the quantum resources an honest prover needs is unbounded in the number of classical bits $n$ involved in the protocol, something not achieved in previous work. Furthermore, this also works in the presence of noise, hence providing us with a practical protocol that will provide security position verification under a quantum bounded storage assumption. We therefore believe that our work provides a blue print to the near-term realization of a new cryptographic primitive, which has the possibility to enhance our communication infrastructure with verified location as an additional security token.

In order to understand the open questions emerging from this work, note that it is important in our analysis of the random protocol that the Boolean function $f$ we choose in order to run the protocols has to be truly random. This implies that the classical circuit to compute $f$ is of exponential size in $n$. To decrease the classical resources needed for this protocol, it is therefore highly relevant to know whether it is possible to use pseudo-randomness instead, or whether there is another way to choose $f$ with a circuit of polynomial size.

Finally, the most important open question is the following: When considering the dependence on the number of classical bits $n$, our lower bound implies that a number of qubits proportional to the number of classical bits sent by the verifiers is needed to attack the scheme.
However, the best construction for a general attack takes $2^n$ EPR pairs \cite{Beigi2011A, Buhrman2013A}. This leaves open the possibility that it could be even harder for attackers to break the security. Can we improve the lower bound to be exponential in $n$?

\section*{Methods}
To prove our main result, we build on the proof strategy used in \cite{Buhrman2013A}, overcoming both conceptual and technical difficulties. For simplicity, we will describe the security proof of the routing protocol first and comment on the differences for the measuring protocol at the end of the section. First, we observe that the joint quantum state of the attackers before their mutual communication arrives already suffices to determine where the qubit will be routed to in the given attack. We subsequently discretize the possible quantum strategies of the attackers with the help of  $\varepsilon$-nets. Since the number of qubits of the attackers is bounded, the size of the $\varepsilon$-nets is limited. From there, we construct classical rounding functions which capture the essentials of the quantum strategies. In particular, an $(\varepsilon,q)$-classical rounding gives rise to a Boolean function for each attack Alice and Bob could do controlling at most $q$ qubits each. These functions agree with the Boolean function $f$ used in the routing protocol on all pairs of classical bit strings $(x,y)$ on which the attackers succeed with probability at least $1-\varepsilon^2$. In this sense, the classical rounding captures the information where the qubit is routed to during an attack. The $\varepsilon$-net construction shows that for $\varepsilon$ small enough and $q \in \mathbb N$, there exists an $(\varepsilon, q)$-classical rounding of size exponential in $q$. For the exact definition of an $(\varepsilon, q)$-classical rounding and the details of the proofs, we refer the reader to the Supplementary Information (see Section \ref{sec:routing} for the routing protocol and Section \ref{sec:measure} for the measuring protocol).

A counting argument that compares the number of $(\varepsilon, q)$-classical roundings to the number of Boolean functions $f$ (on $2n$ bits) used to define the protocol then shows that most Boolean functions are far from any functions produced from classical roundings. More precisely, we show that for $q \leq n/2 - 5$, there exists a function $f:\{0,1\}^{2n} \to \{0,1\}$ that agrees with any function produced from the $(\varepsilon, q)$-classical rounding constructed previously on less than $3/4$ of the possible input pairs $(x,y)$. Moreover, a uniformly random function $f$ has this property with probability at least $1-2^{-2^{n}}$.

Picking $f$ as above, the main result can then be proven by contradiction from the properties of an $(\varepsilon, q)$-classical rounding. Indeed, the counting argument implies that attackers controlling at most $n/2- 5$ qubits each have to be detected with probability greater than $\varepsilon^2$ on at least $1/4$ of all possible pairs of bit strings $(x,y)$. This shows that cheaters will be detected with constant probability for a random function. 

For the measuring protocol, we can show that similarly the attackers have to decide in which basis to measure the qubit $Q$ already before their mutual communication. The argument is based on an entropic uncertainty relation relative to quantum side information \cite{RB09A,BCC+10A}. The rest of the proof proceeds as for the routing protocol.

To obtain lower bounds for concrete functions, we consider the distributional communication complexity in the simultaneous message passing model. Here, Alice and Bob receive inputs $x$ and $y$ and send each a message of equal length to a referee. The latter is supposed to compute the value of the function with probability at least $3/4$ if the inputs are drawn from the uniform distribution. The aforementioned communication complexity is the number of bits Alice (or Bob) has to send. We prove that the routing and the measuring protocols are secure for a function with communication complexity at least $k$ against attackers that control at most $\log(k)/2 - 3$ qubits each. The key insight is that any $(\varepsilon,q)$-classical rounding can be converted into a protocol in the communication complexity setting. Since successful attack strategies lead to $(\varepsilon,q)$-classical roundings, the number of qubits $q$ of the attackers cannot be too small, since otherwise very efficient communication protocols would exist, contradicting the lower bound on the communication complexity. 

\section*{Acknowledgements}

The authors would like to thank Adrian Kent for organizing a workshop on relativistic quantum information theory in February 2020 during which part of this work was presented. AB and MC acknowledge financial support from the European Research Council (ERC Grant Agreement No. 81876), VILLUM FONDEN via the QMATH Centre of Excellence (Grant No.10059) and the QuantERA ERA-NET Cofund in Quantum Technologies implemented within the European Union’s Horizon 2020 Programme (QuantAlgo project) via the Innovation Fund Denmark.

\newpage

\section*{\bf{Supplementary material}}

\section{Main results} \label{sec:contribution}

In this supplementary material, we provide proofs for all results mentioned in the main text. In particular, we prove the following theorems, which are our main results:

Our first result is that the routing protocol \PVroute is secure if Alice and Bob control less than $n/2 - 5$ qubits, where $2n$ classical bits are being sent.  
\begin{thm} \label{thm:intro}
 Let $n \geq 10$. Let us assume that the verifiers choose the bit strings $x$, $y$ of length~$n$ uniformly at random. Then there exists a function $f:\{0,1\}^{2n} \to \{0,1\}$ with the property that, if the number $q$ of qubits each of the attackers controls satisfies 
\begin{equation*}
q \leq \frac{1}{2}n- 5,
\end{equation*}
the attackers are caught during \PVroute with probability at least $2 \cdot 10^{-2}$. Moreover, a uniformly random function $f$ will have this property (except with exponentially small probability).
\end{thm}
Theorem \ref{thm:intro} follows from Corollary \ref{cor:M2-repetitions} in this supplementary material. Already in the original publication of Kent, Munro, and Spiller \cite{Kent2011} which proposed the routing protocol, it was shown that it is possible to attack this scheme if attackers share $2^n$ EPR pairs.
Buhrman, Fehr, Schaffner, and Speelman~\cite{Buhrman2013} studied this class of protocols further, introducing the garden-hose model of communication complexity, which captures attacks relying on teleportation, and showed that an attack exists on the routing protocol using at most $\mathrm{GH}(f)$ EPR pairs. 
Here, $\mathrm{GH}(f)$ is the garden-hose complexity of the function $f$, a measure which is at most polynomial if the function is computable by a log-space Turing machine, but is exponential for a random function.

Our second result is that the measuring protocol \PVBB is also secure if Alice and Bob control less than $n/2 - 5$ qubits, where $2n$ classical bits are being sent. 

\begin{thm} \label{thm:intro-2}
 Let $n \geq 10$. Let us assume that the verifiers choose the bit strings $x$, $y$ of length~$n$ uniformly at random. Then there exists a function $f:\{0,1\}^{2n} \to \{0,1\}$ with the property that, if the number $q$ of qubits each of the attackers controls satisfies 
\begin{equation*}
q \leq \frac{1}{2}n- 5,
\end{equation*}
the attackers are caught during \PVBB with probability at least $2 \cdot 10^{-2}$. Moreover, a uniformly random function $f$ will have this property (except with exponentially small probability).
\end{thm}
Theorem \ref{thm:intro-2} follows from Theorem \ref{thm:second-main-result}. The results of \cite{Buhrman2013} can be adapted to construct attacks on \PVBB{} for which the entanglement required is given by the garden-hose complexity of $f$ (a measure that is polynomial for log-space functions, but can be exponential in general). Thus, a general attack on \PVBB{} is possible when attackers share $2^n$ EPR pairs for any function $f$. As an aside, we do note that, despite the fact that these specific attacks can be translated, we do not know in general whether an attack on \PVroutemod can be translated into an attack on \PVBB{} or vice-versa.

Finally, we consider concrete instead of random functions $f$. In particular, for the binary inner product function\begin{equation} \label{eq:ip}
IP(x,y)=\sum^n_{i=1} x_i y_i \pmod{2} \,,
\end{equation}
we can prove the following:
\begin{thm}\label{thm:intro-IP-lower-bounds}
Let $n \geq 10$. Let us assume that the verifiers choose the bit strings $x$, $y$ of length~$n$ uniformly at random. If the number $q$ of qubits each of the attackers controls satisfies 
\begin{equation*}
    q\leq \frac{1}{2}\log n - 5,
\end{equation*}
the attackers are caught during $\widetilde{PV}_{\mathrm{route}}^{IP}$ and $PV_{\mathrm{meas}}^{IP}$ with probability at least $2 \cdot 10^{-2}$, respectively.
\end{thm}
The statement follows from Theorem \ref{thm:comm-complex}. The supplementary material is organized as follows: Section \ref{sec:prelim} contains some preliminaries concerning communication matrices and the purified distance between quantum states. Our results concerning the routing protocol appear in Section \ref{sec:routing}. Subsequently, we consider the measuring protocol in Section \ref{sec:measure}, before we prove both protocols to be noise robust in Section \ref{sec:noise}. Lower bounds for concrete instead of random functions for both protocols are proven in Section \ref{sec:concrete-functions}. In Section \ref{sec:attack-model} we discuss the importance of the attack model in results on quantum position verification and compare our results to previous and independent work. Finally, we conclude in Section \ref{sec:technical} with some technical results which are needed in the proofs.

\section{Preliminaries} \label{sec:prelim}

\subsection{Communication matrix}
Let $d_H:\{0,1\}^n \times \{0,1\}^n \to \mathbb N$ be the Hamming distance. Let us define for $a$, $n \in \mathbb N$,
\begin{equation*}
V(n,a)= \sum_{l = 0}^a \binom{n}{l}.
\end{equation*}
That is the cardinality of the ball of Hamming distance $a$. Let $\lambda \in (0,1/2)$ be such that $\lambda n \in \mathbb N$. In \cite[p.310]{MacWilliams1977}, we find the useful bound
\begin{equation}\label{eq:entropy_estimate}
V(n,\lambda n) \leq 2^{nh(\lambda)},
\end{equation}
where $h(p) := - p \log p - (1-p) \log (1-p)$ is the binary entropy function. The function $\log$ will be the logarithm with respect to base $2$ in this paper. 

Let $f:\{0,1\}^n \times \{0,1\}^n \to \{0,1\}$. The \emph{communication matrix} of $f$ is defined as
\begin{equation*}
(M_f)_{x,y} = f(x,y).
\end{equation*}
It is thus a $2^n \times 2^n$ matrix. The Hamming distance $d_H(M_f, M_g)$ therefore tells you, for how many pairs of bit strings $(x,y)$ the value $g(x,y)$ differs from $f(x,y)$. Note that here we interpret $M_f$, $M_g$ as strings of length $2^{2n}$.

\subsection{Fidelity and purified distance}
Let us define the fidelity between two quantum states as
\begin{equation} \label{eq:fidelity}
F(\rho, \sigma) := \tr[\sqrt{\sqrt{\sigma}\rho \sqrt{\sigma}}].
\end{equation}
In particular, $F(\ket{\psi}, \ket{\phi}) = |\braket{\psi|\phi}|$. Here, we write $F(\ket{\psi}, \ket{\varphi})$ for pure states $\ket{\psi}$, $\ket{\varphi}$ to mean $F(\dyad{\psi}, \dyad{\varphi})$ for brevity. Note that sometimes the fidelity is defined as the square of \eqref{eq:fidelity}. 
The fidelity can be used to define the purified distance on the set of density matrices \cite[Definition 3.8]{Tomamichel2016}. For quantum states $\rho$, $\sigma$, it is defined as
\begin{equation*}
\mathcal P(\rho, \sigma) := \sqrt{1 - F(\rho, \sigma)^2}.
\end{equation*}
Again, we often write $\mathcal P(\ket{\psi}, \ket{\varphi})$ for pure states $\ket{\psi}$, $\ket{\varphi}$ instead of $\mathcal P(\dyad{\psi}, \dyad{\varphi})$. Unlike the fidelity, the purified distance is a metric on the set of states, which makes it easier to work with (see e.g.\ {\cite[Proposition 3.3]{Tomamichel2016}}). In particular, it satisfies the triangle ineguality. 

\section{The qubit routing protocol} \label{sec:routing}

\subsection{Qubit routing} \label{sec:routing-protocol}

Let $f:\{0,1\}^{2n} \to \{0,1\}$. We consider PV in one spatial dimension. We will start by considering a modified version of the routing protocol which uses a maximally entangled pair. Our main result will follow by realizing in Section \ref{sec:alternative} that the entangled and unentangled qubit routing protocols are essentially equivalent. The general setup for the entangled qubit routing protocol \PVroutemod is the following: The prover $P$ claims to be at position $z$ on a line. To the left and right of $z$ are the verifiers $V_0$ and $V_1$. All communication happens at the speed of light. The protocol \PVroutemod considered in \cite{Buhrman2013} is the following (see Figure 2 of that paper):
\begin{enumerate}
\item $V_0$ randomly chooses two $n$-bit strings $x$, $y$, computes $f(x,y)$ and sends $y$ on to $V_1$. Moreover, $V_0$ prepares a maximally entangled $2$-qubit state $\ket{\Omega}=\frac{1}{\sqrt{2}}(\ket{00} + \ket{11})$. If $f(x,y) = 0$, $V_0$ does nothing, if $f(x,y) = 1$, $V_0$ sends one qubit $R$ of $\ket{\Omega}$ to $V_1$.
\item $V_0$ sends the other qubit $Q$ of $\ket{\Omega}$ together with $x$ to $P$. $V_1$ sends $y$ such that it arrives at $z$ at the same time as $x$ and $Q$ sent by $V_0$.
\item $P$ sends the qubit $Q$ on to $V_{f(x,y)}$.
\item $V_0$ and $V_1$ accept if the qubit arrives at the correct time at the correct verifier and a Bell measurement of both qubits yields the correct outcome. 
\end{enumerate}
The timing of the response from $P$ is deemed correct if it is compatible with the qubit $Q$ originating from $z$ right after it reached that point. An illustration of the protocol can be found in Figure 2 in the main text.

 The advantage of this protocol compared to others is that the honest prover only needs to handle one qubit. Note that this qubit could even be presented as a logical qubit in an error-correcting code to combat noise in the communication line. Since the protocol only requires routing this qubit and no further processing of it, any error correcting code (even without fault-tolerant properties) is fine. Errors in creating the qubit states and verifying it on the side of the  verifiers, however, need to be carried out in a fault-tolerant manner.

We will now give the form of the most general attack on \PVroutemod. Note that we can restrict our attention to unitaries by considering the Stinespring dilation of the quantum channels the attackers might wish to perform. There are two attackers Alice and Bob, where Alice is between $V_0$ and $z$ and Bob is between $z$ and $V_1$. However, neither of the attackers is actually at $z$. As explained in Section \ref{sec:contribution}, the verifiers hold a qubit system $R$, while Alice holds a qubit system~$A$, a local quantum system~$\tilde A$ and a quantum system used for communication~$A_c$. Bob has similar systems $B$, $\tilde B$ and $B_c$. 
\begin{defi}[$q$-qubit strategy for \PVroutemod] \label{defi:strategy}
Fix a partition into systems $R A \tilde A A_c B \tilde B B_c$. Both Alice's and Bob's registers each consist of $q$ qubits.
Let $d$ be the combined dimension of this system, therefore $d = 2^{2q+1}$. A \emph{$q$-qubit strategy for \PVroutemod} consists of the starting state $\ket{\psi}$ on $R A \tilde A A_c B \tilde B B_c$ and of unitaries $U^x_{A \tilde A A_c}$, $V^y_{B \tilde B B_c}$, $K^{xy}_{A \tilde A B_c}$ and $L^{xy}_{B \tilde B A_c}$ for all $x$, $y \in \{0,1\}^n$. The superscripts indicate whether the unitaries may depend on the message $x$ that $V_0$ sends, the message $y$ that $V_1$ sends or on both messages.
\end{defi}
\noindent Note that Alice and Bob only hold equally many qubits at the beginning of the strategy, after the communication phase the numbers can be different. An illustration of the above can be found in Figure 4 in the main text.

\begin{remark}\label{rmk:initial-state-product}
From the description of the protocol, it is clear that we only need consider strategies for which $\ket{\psi} = \ket{\Omega}_{AR} \otimes \ket{\psi^\prime}_{\tilde A A_c B \tilde B B_c}$ to prove the protocol secure. However, it is advantageous to consider more general starting states: In reality, photon signals over fiber travel at around $(2/3)c$, where $c$ is the speed of light, while we want to allow our attackers to signal at speed~$c$.  If our proof can quantify over all pre-shared states between $R A \tilde A A_c B \tilde B B_c$, then conceptually the input message can be `slow'. We could even imagine the state being available long before the protocol, with Alice and Bob distributing the state amongst themselves however they want. The input timing is then only on the classical messages. Alternatively, they could start computing locally before the classical messages $x$ and $y$ are available. All these scenarios lead to a starting state not of the form $\ket{\Omega} \otimes \ket{\psi^\prime}$, which is why considering general states $\ket{\psi}$ within the 
$q$-qubit strategies for \PVroutemod only makes the security notion stronger.
\end{remark}

\begin{remark}\label{rmk:random-does-not-help}
It can easily be seen that shared randomness between Alice and Bob does not help them for a fixed function $f$. Indeed, if $\rho$ is the reduced state at the end of the protocol on $RA$ if $f(x,y) = 0$ or $RB$ if $f(x,y) = 1$, the probability that Alice and Bob are not caught by the verifiers is $\bra{\Omega}\rho\ket{\Omega}$. Note that the objective function $\bra{\Omega}\rho\ket{\Omega}$ is linear in $\rho$ and the partial trace is a linear map. Thus, the maximum over convex combinations of strategies is achieved at deterministic strategies $\{U^x, V^y, K^{xy}, L^{xy}\}_{xy}$.
\end{remark}
\noindent The main lower-bound result of \cite{Buhrman2013} concerning the entangled qubit routing protocol \PVroutemod is the following:
\begin{thm}[{\cite[Theorem E.4]{Buhrman2013}}] \label{thm:E4}
Let $q$, $n \in \mathbb N$. For any $q$-qubit starting state $\ket{\psi}$ on $R A \tilde A A_c B \tilde B B_c$, there exists a Boolean function on inputs $x$, $y \in \{0,1\}^n$ such that any perfect attack on \PVroutemod requires $q$ to be linear in $n$.
\end{thm}

On the one hand, this theorem proves that \PVroutemod is secure in some sense if the number of qubits the attackers control is at most linear in $n$. On the other hand, it has several features which make it unsuitable to derive any limits for actual attacks on the \PVroutemod-scheme from it.
Firstly, it only discusses perfect attacks, while actual attackers would still be practically successful if they have a small probability of being caught.
Secondly, the theorem fixes the state before quantifying over the functions, while actual attackers would be able to choose their entanglement after knowing the function $f$. This can be interpreted as a violation of Kerkhoffs’s principle, since the function $f$ must not be known to the attackers beforehand.
Finally, the theorem only shows that there exist an input pair $x,y$ for which the attackers will be detectable, but does not say anything about how many such pairs exist, leaving the possibility that these pairs might only be asked with exponentially small probability.  These severe drawbacks make Theorem \ref{thm:E4} unsuitable for practical applications.

The aim of this work is therefore to improve upon Theorem \ref{thm:E4} and to provide a version that solves all three problems, thus recovering the statement under a more realistic class of attacks.

\subsection{Lower bounds on the entangled qubit routing protocol}

We start our analysis of the entangled qubit routing protocol by defining an $(\epsilon,l)$-perfect $q$-qubit strategy as one which has a high chance of being accepted by the verifiers at the end of the protocol. In~\cite{Buhrman2013}, only perfect strategies were considered, i.e.\ $\epsilon = 0$, $l = 2^{2n}$. Moreover, we want to allow that the attackers only succeed on $l$ of the $2^{2n}$ pairs of bit strings $(x,y)$ that the verifiers might send.
\begin{defi}[$(\epsilon, l)$-perfect $q$-qubit strategy for \PVroutemod]
Let $\epsilon > 0$, $l \in \mathbb N$. A  $q$-qubit strategy for \PVroutemod as in Definition~\ref{defi:strategy} is \emph{$(\epsilon, l)$-perfect} if on $l$ pairs of strings $(x,y)$, Alice and Bob are caught by the verifiers with probability at most $\epsilon^2$.  
\end{defi}

\begin{remark}
Note that Alice and Bob are caught by the verfiers on input $(x,y)$ with probability at most $\epsilon^2$ if and only if Alice and Bob produce a state $\ket{\tilde \psi}$ at the end of the protocol such that $\mathcal P(\rho_{RA}, \dyad{\Omega}_{RA}) \leq \epsilon$ if $f(x,y) = 0$ and $\mathcal P(\rho_{RB}, \dyad{\Omega}_{RB}) \leq \epsilon$ if $f(x,y) = 1$, where $\rho$ is the corresponding reduced state of $\ket{\tilde \psi}$.
\end{remark}
The following proposition relates the above definition to the purified distance with respect to the state $\ket{\tilde \psi}$ as in \cite[Appendix E]{Buhrman2013}. It is a direct consequence of Uhlmann's theorem \cite[Theorem 3.22]{Watrous2018}.
\begin{prop} \label{prop:F_reduced}
For a state $\ket{\tilde \psi}_{RA\tilde A A_c B \tilde B B_c}$, it holds that
\begin{equation*}
\inf_{\ket{\phi}} \mathcal P( \ket{\tilde \psi}_{RA\tilde A A_c B \tilde B B_c}, \ket{\Omega}_{RA} \otimes \ket{\phi}_{\tilde{A}A_cB\tilde{B}B_c}) = \mathcal  P(\rho_{RA}, \dyad{\Omega}_{RA})
\end{equation*}
and
\begin{equation*}
\inf_{\ket{\phi}} \mathcal P( \ket{\tilde \psi}_{RA\tilde A A_c B \tilde B B_c}, \ket{\Omega}_{RB} \otimes \ket{\phi}_{A\tilde{A}A_c\tilde{B}B_c}) =  \mathcal P(\rho_{RB}, \dyad{\Omega}_{RB}),
\end{equation*}
where $\rho_{RA}$ and $\rho_{RB}$ are the corresponding reduced density matrices of $\ket{\tilde \psi}_{RA\tilde A A_c B \tilde B B_c}$.
\end{prop}

Before we go on, we define the sets of states from which the routed qubit can be recovered by attacker Alice or Bob, to be returned to $V_0$ or $V_1$ respectively. Note that we will always write $A_X$ instead $A_X \otimes I_{X^c}$ for ease of notation, where $X$ is a system with complement $X^c$, $A$ an operator and $I$ the identity operator. 
\begin{defi} \label{defi:good-sets-routing}
Let $ \epsilon \in [0,1]$. We define $\mathcal S_0^{\epsilon, \mathrm{route}}$ as the set of states $\ket{\varphi}_{RA\tilde A A_c B \tilde B B_c}$ for which there exists a unitary $K_{A \tilde A B_c}$ such that $\mathcal P(\rho_{RA}, \dyad{\Omega}_{RA}) \leq \epsilon$, where $\rho$ is the reduced state of $K \ket{\varphi}$. Moreover, we define $\mathcal S_1^{\epsilon, \mathrm{route}}$ as the set of states $\ket{\varphi^\prime}_{RA\tilde A A_c B \tilde B B_c}$ for which there exists a unitary $L_{B \tilde B A_c}$ such that $\mathcal P(\rho^\prime_{RB}, \dyad{\Omega}_{RB}) \leq \epsilon$, where $\rho^\prime$ is the reduced state of $ L \ket{\varphi^\prime}$.
\end{defi}

Now, we consider a state that can be used to reveal the qubit at $V_0$ in the last step of a $q$-qubit strategy and a state that can be used to reveal the qubit at $V_1$ in the last step of the strategy. We prove a proposition which formalizes the idea that these two such states have to differ by at least a certain amount. This shows that the sets we just defined are disjoint if we choose $\epsilon$ small enough.
This proposition can be seen as a robust version of \cite[Lemma E.1]{Buhrman2013}.

\begin{prop} \label{prop:states-must-have-low-fidelity}
Let $0 \leq \epsilon \leq 0.41$ and let $\ket{\psi_0}_{RA\tilde A A_c B \tilde B B_c} \in \mathcal S_0^{\epsilon, \mathrm{route}}$, $\ket{\psi_1}_{RA\tilde A A_c B \tilde B B_c} \in \mathcal S_1^{\epsilon, \mathrm{route}}$. Then,
\begin{equation*}
    \mathcal P(\ket{\psi_0}, \ket{\psi_1}) > 0.046.
\end{equation*}
\end{prop}

\begin{proof}
By the definition of the sets in Definition \ref{defi:good-sets-routing}, there exist unitaries $K_{A \tilde A B_c}$ and $L_{B \tilde B A_c}$ such  that $\mathcal P(\rho_0, \dyad{\Omega}_{RA} ) \leq \epsilon$ and $\mathcal P(\rho_1, \dyad{\Omega}_{RB} ) \leq \epsilon$ for $\rho_0$ the reduced state on $RA$ of $K_{A \tilde AB_c} \ket{\psi_0}$, $\rho_1$ the reduced state on $RB$ of $L_{B \tilde B A_c} \ket{\psi_1}$. 
By Proposition \ref{prop:F_reduced} and compactness, we can find states $\ket{\phi_0}_{\tilde{A}A_cB\tilde{B}B_c}$ and $\ket{\phi_1}_{A\tilde{A}A_c\tilde{B}B_c}$ such that
\begin{align*}
    \mathcal P(K_{A \tilde AB_c} \ket{\psi_0}_{RA\tilde{A}A_cB\tilde{B}B_c}, \ket{\Omega}_{RA} \otimes \ket{\phi_0}_{\tilde{A}A_cB\tilde{B}B_c}) &= \mathcal P(\rho_0, \dyad{\Omega}_{RA} ), \\
    \mathcal P(L_{B \tilde B A_c} \ket{\psi_1}_{RA\tilde{A}A_cB\tilde{B}B_c}, \ket{\Omega}_{RB} \otimes \ket{\phi_1}_{A\tilde{A}A_c\tilde{B}B_c}) &=    \mathcal P(\rho_1, \dyad{\Omega}_{RB} ).
\end{align*}
Applying the triangle inequality twice and using the fact that $\mathcal P(U \ket{\varphi}, U \ket{\psi})=\mathcal P(\ket{\varphi}, \ket{\psi})$ for any unitary $U$ and any states $\ket{\varphi}$, $\ket{\psi}$, we obtain
\begin{equation*}
\mathcal P(\ket{\psi_0}, \ket{\psi_1}) \geq \mathcal P(K^\ast \ket{\Omega} \otimes \ket{\varphi_0}, L^\ast \ket{\Omega} \otimes \ket{\phi_1}) - \mathcal P(K \ket{\psi_0},  \ket{\Omega} \otimes \ket{\varphi_0}) - \mathcal P(L \ket{\psi_1}, \ket{\Omega} \otimes \ket{\varphi_1}).
\end{equation*}
We can estimate the last two terms on the right hand side as
\begin{align*}
\mathcal P(K \ket{\psi_0},  \ket{\Omega} \otimes \ket{\varphi_0}) & \leq \epsilon,\\
\mathcal P(L \ket{\psi_1}, \ket{\Omega} \otimes \ket{\varphi_1}) & \leq \epsilon.
\end{align*}
These inequalities hold by assumption. By the computations of \cite[Lemma E.1]{Buhrman2013} (repeated as Lemma \ref{lem:E1} for completeness), we can estimate the first term as
\begin{equation*}
    \mathcal P(K^\ast \ket{\Omega} \otimes \ket{\varphi_0}, L^\ast \ket{\Omega} \otimes \ket{\phi_1}) \geq \frac{\sqrt{3}}{2}.
\end{equation*}
Thus, 
\begin{equation*}
\mathcal P(\ket{\psi_0},\ket{\psi_1}) \geq \frac{\sqrt{3}}{2} - 2 \epsilon
\end{equation*}
and the assertion follows using the assumption $\epsilon \leq 0.41$.
\end{proof}

We will need a final easy lemma to convert between Euclidean distance and purified distance. It is a direct consequence of the fact that $1-x^2=(1-x)(1+x) \leq 2(1-x)$ for $x \in [0,1]$.
\begin{lem}\label{lem:euclidean_ball_fidelity}
Let $\ket{x}$, $\ket{y} \in \mathbb C^d$ be two unit vectors. Then,
\begin{equation*}
\mathcal P(\ket{x},\ket{y}) \leq \|\ket{x}-\ket{y}\|_2.
\end{equation*}
\end{lem}

We observe that Proposition~\ref{prop:states-must-have-low-fidelity} implies that the attackers in some sense already decide before their communication step where the qubit can end up at the end of the protocol.
Therefore, if the dimension of the state they share is small enough, a classical description of the first part of their strategy yields a compression of $f$.
The classical compression is captured in the following notion of classical roundings:
 \begin{defi}[$(\epsilon, q)$-classical rounding] \label{def:classical-rounding}
Let $q$, $k$, $n \in \mathbb N$, $\epsilon > 0$. Then, 
\begin{equation*}
g:\{0,1\}^{3k} \to \{0,1\}    
\end{equation*}
is an \emph{$(\epsilon, q)$-classical rounding} of size $k$ if for all $f:\{0,1\}^{2n} \to \{0,1\}$, for all states $\ket{\psi}$ on $2q+1$ qubits, for all $l \in \{1, \ldots, 2^{2n}\}$ and for all $(\epsilon, l)$-perfect $q$-qubit strategies for \PVroutemod, there are functions $f_A:\{0,1\}^n \to \{0,1\}^{k}$, $f_B:\{0,1\}^n \to \{0,1\}^{k}$ and $\lambda \in \{0,1\}^{k}$ such that $g (f_A(x), f_B(y), \lambda) =f(x,y)$ on at least $l$ pairs $(x,y)$.
\end{defi}

The function $\tilde f$ defined as
\begin{equation} \label{eq:f-tilde}
  \tilde f(x,y) := g(f_A(x), f_B(y), \lambda)   \qquad \forall x, y \in \{0,1\}^n
\end{equation}
in a classical rounding hence measures how good the $q$-qubit strategy Alice and Bob use performs for the qubit routing specified by the function $f$. For example, if the strategy is an $(\epsilon, 2^{2n})$-perfect $q$-qubit strategy, then $f=\tilde f$.

Since the following statement holds for both the routing and the measure protocol, which we consider in the next section, we prove it here for both protocols, although the sets $\mathcal S_i^{\epsilon, \mathrm{meas}}$ are only defined in Definition \ref{def:good-sets-BB84}.

\begin{lem} \label{lem:classical_compression}
Let $\# \in \{\mathrm{route}, \mathrm{meas}\}$, $q \in \mathbb N$. Furthermore, let $0\leq\epsilon \leq \epsilon_0$, where $\epsilon_0$ is such that $\ket{\varphi_i} \in \mathcal S_i^{\epsilon, \#}$, $i \in \{1,2\}$ implies $\mathcal P(\ket{\varphi_0}, \ket{\varphi_1}) > 0.013$. 
Then, there is an $(\epsilon, q)$-classical rounding of size $\log(927) 2^{2q+2}$. 
\end{lem}
\begin{proof}
 We consider $\delta = 0.00216$. Let us choose a $\delta$-net $\mathcal N_S$ in Euclidean norm for the set of pure states on $2q+1$ qubits, where the net has cardinality at most $2^{k}$. Likewise, let us choose $\delta$-nets $\mathcal N_A$ and $\mathcal N_B$ in operator norm for the set of unitaries in dimension $2^{q}$, where the nets have cardinalities at most $2^k$ each. We will show at the end of the proof that we can choose $k$ as in the assertion.
 
Let us now construct the $(\epsilon, d)$-classical rounding $g$ as in Definition \ref{def:classical-rounding}. Let $x^\prime \in \{0,1\}^{k}$, $y^\prime \in \{0,1\}^{k}$ and $\lambda \in \{0,1\}^{k}$ and let $U \in \mathcal N_A$ be the element with index $x^\prime$, $V \in \mathcal N_B$ be the element with index $y^\prime$ and $\ket{\phi} \in \mathcal N_S$ be the element with index $\lambda$. Then, we define $g(x^\prime, y^\prime, \lambda) = 0$ if $U_{A \tilde A A_c} \otimes V_{B \tilde B B_c} \ket{\phi}$ is closer to $\mathcal S_0^{\epsilon, \#}$ than to $\mathcal S_1^{\epsilon, \#}$ in purified distance and $g(x^\prime, y^\prime, \lambda) = 1$ if $U_{A \tilde A A_c} \otimes V_{B \tilde B B_c} \ket{\phi}$ is closer to $\mathcal S_1^{\epsilon, \#}$ than to $\mathcal S_0^{\epsilon, \#}$ in purified distance. If neither is the case, we make the arbitrary choice $g(x^\prime, y^\prime, \lambda) = 1$. Since the assumption on $\epsilon_0$ implies that $\mathcal S_0^{\epsilon, \#} \cap \mathcal S_1^{\epsilon, \#} = \emptyset$, this is a well-defined function.
 
 It remains to show that $g$ is indeed an $(\epsilon, q)$-classical rounding. We consider an arbitrary $f:\{0,1\}^{2n} \to \{0,1\}$ and an arbitrary state on $2q+1$ qubits $\ket{\psi}$. Let $\ket{\psi}$ and $\{U^x, V^y\}_{xy}$ be from a $q$-qubit strategy for $PV_{\#}^f$. We choose $\lambda$ as the index of the closest element from $\mathcal N_S$ to $\ket{\psi}$ in Euclidean norm. Moreover, we choose $f_A(x)$ to be the closest element from $\mathcal N_A$ to $U^x$ in operator norm and $f_B(y)$ to be the closest element from $\mathcal N_B$ to $V^y$ in operator norm. If the closest element is not unique, we make an arbitrary choice. We claim that if $U^x \otimes V^y \ket{\psi} \in \mathcal S^{\epsilon, \#}_0$, then $U \otimes V \ket{\phi}$ is closer to $\mathcal S^{\epsilon, \#}_0$ than to $\mathcal S^{\epsilon, \#}_1$, where $(U, V, \ket{\phi})$ are the elements from the nets corresponding to $(f_A(x), f_B(y), \lambda)$. In particular, we will show that for $\delta$ as chosen above,
 \begin{equation} \label{eq:still-close}
     \mathcal P(U^x \otimes V^y \ket{\psi}, U \otimes V \ket{\phi}) < 0.0065.
 \end{equation}
 Since $\mathcal P(\ket{\psi_0}, \ket{\psi_1}) > 0.013$  for $\ket{\psi_0} \in \mathcal S_0^{\epsilon,\#}$ and $\ket{\psi_1} \in \mathcal S_1^{\epsilon,\#}$, the claim then follows by the triangle inequality for the purified distance. Thus, we now prove \eqref{eq:still-close}. Let $\Delta_A := U^x - U$, $\Delta_B := V^y - V$ and $\ket{\Delta_S} = \ket{\psi} - \ket{\phi}$. Note that $\|\Delta_A\|_\infty \leq \delta$, $\|\Delta_B\|_\infty \leq \delta$ and $\|\ket{\Delta_S}\|_2 \leq \delta$. Indeed, using Lemma~\ref{lem:euclidean_ball_fidelity}, 
 \begin{align*}
     \mathcal P(U^x \otimes V^y \ket{\psi}, U \otimes V \ket{\phi}) & \leq \|U^x \otimes V^y \ket{\psi} - U \otimes V \ket{\phi}\|_2 \\
     & \leq \|(U + \Delta_A) \otimes (V + \Delta_B) (\ket{\phi} + \ket{\Delta_S}) - U \otimes V \ket{\phi}\|_2\\
     & \leq 3 \delta + 3 \delta^2 + \delta^3.
 \end{align*}
 In the last line, we have used the triangle inequality together with $\|X \otimes Y\ket{\eta}\|_2 \leq \|X\|_\infty \|Y\|_\infty \|\ket{\eta}\|_2$. For $\delta = 0.00216$, we can compute $3 \delta + 3 \delta^2 + \delta^3 < 0.0065$ and  \eqref{eq:still-close} follows.
 
 Finally, consider an $(\epsilon, l)$-perfect strategy for $PV^f_{\#}$ and let $(x, y)$ be such that the attackers are caught with probability at most $\epsilon^2$. Without loss of generality, let $(x, y)$ be such that $f(x,y) = 0$. Then, it holds in particular that $U^x \otimes V^y \ket{\psi} \in \mathcal S_0^{\epsilon, \#}$. Thus, using \eqref{eq:still-close}, it follows that $g(f_A(x), f_B(y), \lambda) = f(x,y)$ on such a pair $(x ,y)$. Since there are at least $l$ pairs $(x ,y)$ which achieve low detection probability for an $(\epsilon, l)$-perfect $q$-qubit strategy, $g(f_A(x) ,f_B(y), \lambda) = f(x,y)$ on at least $l$ pairs $(x,y)$. Hence, $g$ is an $(\epsilon, q)$-classical rounding.
 
 In order to conclude the proof, we must estimate $k$. Lemma 9.5 of \cite{ledoux1991probability} implies that $\mathcal N_A$, $\mathcal N_B$, $\mathcal N_S$ can be chosen to have cardinality at most
 \begin{equation*}
     |\mathcal N_A| \leq (927)^{2^{2q+1}}, \qquad |\mathcal N_B| \leq (927)^{2^{2q+1}} \quad \textrm{and} \quad |\mathcal N_S| \leq (927)^{2^{2q+2}}.
 \end{equation*}
 Taking the logarithm, the desired bounds on the size of the classical rounding follow.
\end{proof}

The next statement says that if we fix a number of qubits $q$ and an error $\epsilon$, the attackers will get a large fraction of the inputs $(x,y)$ wrong in any $q$-qubit strategy if we choose $f$ to be random and if the number of qubits of the state $\ket{\psi}$ in the strategy is not too large.
\begin{lem} \label{lem:stable_E4}
Let $\epsilon \in [0,1]$ $n$, $k$, $q \in \mathbb N$, $n \geq 10$. Moreover, fix an $(\epsilon, q)$-classical rounding $g$ of size $k$ with  
    $k=\log(927) 2^{2q+2}$. Let
\begin{equation*}
q \leq \frac{1}{2}n-5.
\end{equation*}
Then, a uniformly random $f: \{0,1\}^{2n} \to \{0,1\}$ fulfills the following with probability at least $1 - 2^{-2^{n}}$:
 For any $f_A:\{0,1\}^n \to \{0,1\}^{k}$, $f_B:\{0,1\}^n \to \{0,1\}^{k}$, $\lambda \in \{0,1\}^{k}$, the equality $g(f_A(x), f_B(y), \lambda) = f(x,y)$ holds on less than $3/4$ of all pairs $(x,y)$.
\end{lem}

\begin{proof}
Let $\tilde f$ be as in \eqref{eq:f-tilde}. Definition \ref{def:classical-rounding} states that given $q \in \mathbb N$ and $\epsilon>0$, the number of functions~$\tilde f$ that Alice and Bob can implement only depends on the number of choices for $f_A$, $f_B$, $\lambda$ (since $g$ is fixed given $\epsilon$ and $q$). Thus, for an $(\epsilon, q)$-classical rounding of size $k \in \mathbb N$, they can implement $2^{ (2^{n+1}+1) k}$ possible functions. Therefore, we want to estimate the probability that for a randomly chosen $f$, we can find $f_A$ and $f_B$ such that the corresponding function $\tilde f$ lies within Hamming distance $1/4\cdot 2^{2n}$ of $f$.
Hence,
\begin{align*}
&\mathbb P(f: \exists f_A, f_B, \lambda \mathrm{~s.t.~} d_H(M_f, M_{\tilde f}) \leq 2^{2n-2}) \\
& = \frac{|\{f: \exists f_A, f_B, \lambda \mathrm{~s.t.~} d_H(M_f, M_{\tilde f}) \leq 2^{2n-2}\}|}{|\{f:\{0,1\}^{2n} \to \{0,1\}\}|} \\
&\leq \frac{|\{f: \exists f_A, f_B, \lambda \mathrm{~s.t.~} f = \tilde f\}|\cdot|V(2^{2n}, 2^{2n-2})|}{|\{f:\{0,1\}^{2n} \to \{0,1\}\}|}\\
&\leq 2^{ (2^{n+1}+1)k} 2^{2^{2n}h(1/4)} 2^{-2^{2n}}
\end{align*}
For the first equality, we use the fact that the function $f$ is drawn uniformly at random. For the first inequality, we estimate the numerator by considering a ball in Hamming distance around every function $\tilde f$ we can express by suitable $f_A$, $f_B$, $\lambda$. In the last line, we have used \eqref{eq:entropy_estimate}. Using $k = \log(927) 2^{2q+2}$ and $q\leq n/2-5$, we infer that $\mathbb P(f: \exists f_A, f_B, \lambda \mathrm{~s.t.~} d_H(M_f, M_{\tilde f}) \leq 2^{2n-2}) $ is strictly bounded from above by $2^{-2^{n}}$. 
\end{proof}

In Lemma \ref{lem:stable_E4}, we have shown that for any state $\ket{\psi}$, a random function $f$ has large Hamming distance to any $\tilde f$ if the dimension of the state is small enough. In particular, this means that any $(\epsilon,3/4\cdot 2^{2n})$-perfect $q$-qubit strategy needs a number of qubits which is linear in the number of classical bits. The following proposition makes this precise.
\begin{prop}\label{prop:after-embezzlement} Let $0 \leq \epsilon \leq 0.41$ and $n \geq 10$, $q$, $n \in \mathbb N$. Then, a uniformly random function $f:\{0,1\}^{2n} \to \{0,1\}$ has the following property with probability at least $1-2^{-2^{n}}$: 
Any $(\epsilon,\frac{3}{4}\cdot 2^{2n})$-perfect $q$-qubit strategy for \PVroutemod requires
\begin{equation}\label{eq:dimension-bound}
q > \frac{1}{2} n - 5,
\end{equation}
where $\ket{\psi}$ is a state on $2q+1$ qubits.
\end{prop}
\begin{proof}
We prove the statement by contradiction. Let $g$ be the $(\epsilon, q)$-classical rounding of size $k$, where $k=\log(927) 2^{2q+2}$,
which is guaranteed to exist by Proposition \ref{prop:states-must-have-low-fidelity} and Lemma \ref{lem:classical_compression}. Assume that $q \leq \frac{1}{2} n - 5$. Pick a function $f: \{0,1\}^{2n}\to \{0,1\}$ such that for any $f_A:\{0,1\}^n \to \{0,1\}^{k}$, $f_B:\{0,1\}^n \to \{0,1\}^{k}$, $\lambda \in \{0,1\}^{k}$ and $\tilde f$ as in \eqref{eq:f-tilde}, the equality 
$f(x,y) = \tilde f(x,y)$ holds on less than $3/4$ of all pairs $(x,y)$. By Lemma \ref{lem:stable_E4}, a uniformly random $f$ will have this property with probability at least $1 - 2^{-2^{n}}$. 

Let $\ket{\psi}$ be a state on $2q+1$ qubits and assume that there is a $(\epsilon,\frac{3}{4}\cdot 2^{2n})$-perfect $q$-qubit strategy for \PVroutemod. Then, the corresponding $f_A$, $f_B$, $\lambda$ satisfy    $g(f_A(x), f_B(y), \lambda) = f(x,y)$ on at least $\frac{3}{4}\cdot 2^{2n}$ pairs $(x,y)$. However, this is a contradiction to the choice of $f$.
\end{proof}

Finally, we can rephrase the previous theorem as a statement about the probability that Alice and Bob are caught by the verifiers. 
\begin{thm} \label{thm:main-result}
Let $f:\{0,1\}^{2n} \to \{0,1\}$, $n \geq 10$ and let 
\begin{equation*}
q \leq \frac{1}{2} n- 5.
\end{equation*}
Let us assume that the verifiers choose a function $f$ uniformly at random before the protocol and that they choose $x$, $y$ uniformly at random during the protocol. Moreover, let us assume that Alice and Bob control at most $q$ qubits each at the beginning of the protocol. Then, the attackers are caught during \PVroutemod with probability at least $4\cdot 10^{-2}$.
\end{thm}
\begin{proof}
Let  $0 <\epsilon \leq 0.41$. By Proposition \ref{prop:after-embezzlement}, with probability at least $1-2^{-2^{n}}$ the function $f$ is such that there are no $( \epsilon,3/4\cdot2^{2n})$-perfect $q$-qubit strategies for \PVroutemod. 
That means that for any strategy Alice and Bob can implement with their state, on a fraction at least $1/4$ of the possible bit strings $(x,y)$, the final reduced state will be at least $\epsilon$ away in purified distance from the maximally entangled state.

That means, that the measurement $\{\dyad{\Omega}, I_4 - \dyad{\Omega}\}$ on such an input pair catches them cheating with probability at least
\begin{equation*}
1 - F(\rho, \dyad{\Omega})^2 > \epsilon^2.
\end{equation*}
Multiplying all these probabilities and using that $n \geq 10$, we obtain the bound in the assertion.
\end{proof}

In order to increase the probability with which the attackers are caught, it is possible to repeat the protocol sequentially several times, as the following proposition shows. It is important to note that Alice and Bob are not allowed to go to $z$, the position of the honest prover, during the time the protocol runs. The implicit assumption of PBC is that $z$ is in some secure zone like a bank, for example, that attackers do not have access to.

\begin{prop} \label{prop:repitition}
Let $f:\{0,1\}^{2n} \to \{0,1\}$, $n \geq 10$, $r$, $q$, $n \in \mathbb N$ and let 
\begin{equation*}
q \leq \frac{1}{2} n- 5.
\end{equation*}
Let us assume that the verifiers choose a function $f$ uniformly at random at the beginning and that they run the protocol \PVroutemod sequentially $r$-times, choosing $x$, $y$ uniformly at random each time. Moreover, let us assume that Alice and Bob control at most $q$ qubits each at the beginning of each iteration of \PVroutemod. Then, the attackers are caught with probability at least $1 - 0.96^r$.
\end{prop}

\begin{proof}
Let $X_i \in \{0,1\}$ be random variables where $i \in \{1,\ldots, r\}$. We set $X_i = 0$ if the attackers are detected in round $i$ and $X_i = 1$ if they are not detected. First, we observe that Theorem \ref{thm:main-result} still holds if the state Alice and Bob share is mixed, because it is equivalent to a random mixture of pure states. Thus, the strategy for a mixed state is a random mixture of strategies for pure states. By Remark \ref{rmk:random-does-not-help}, shared randomness does not increase the probability of the attackers to avoid detecion. Between the repetitions, it can happen that the state Alice and Bob share depends on previous iterations. However, the qubit of the maximally entangled pair at the beginning of each round is uncorrelated with that state and the pair $(x,y)$ in each round does not depend on previous rounds. Moreover, the attackers are assumed to control at most $q$ qubits at the beginning of each round.  Thus, the probability that the attackers are not detected is at most $\mathbb P(X_i = 1|X_{i-1} = x_{i-1}, \ldots,  X_{1} = x_{1}) \leq  0.96$  for any $i\in \{1, \ldots, r\}$ by Theorem \ref{thm:main-result}. Since there are $r$ rounds, the probability to escape detection in all rounds is 
\begin{equation*}
    \mathbb P(X_r = 1, \ldots, X_1 = 1) = \prod_{i = 1}^r \mathbb P(X_i = 1|X_{i-1} = 1, \ldots,  X_{1} = 1) \leq 0.96^r.
\end{equation*}
This proves the assertion.
\end{proof}

\subsection{The qubit routing protocol} \label{sec:alternative}

For ease of analysis, we have been considering a protocol where the verifiers hold a reference qubit, but an almost-equivalent protocol exists where the verifiers only need to store classical information. This is the qubit routing protocol considered in the main text.
Let $\ket{\pm} = \frac{1}{\sqrt{2}}(\ket{0} \pm \ket{1})$. As known from the context of the BB84 protocol \cite{Bennett1984}, we can replace the final measurement  $\{\dyad{\Omega}, I_4 - \dyad{\Omega}\}$ by the following measurement:
With probability $\frac{1}{2}$ each, measure either $\{\dyad{++} + \dyad{--}, I_4-(\dyad{++} + \dyad{--})\}$ or $\{\dyad{00} + \dyad{11}, I_4 - (\dyad{00} + \dyad{11})\}$. We denote this measurement by M2. Let us therefore consider the slightly altered protocol \PVroute:
\begin{enumerate}
\item $V_0$ chooses an $n$-bit string $x$ uniformly at random, $V_1$ chooses an $n$-bit string $y$ uniformly at random. Moreover, $V_0$ prepares one of the states $\ket{0}$, $\ket{1}$, $\ket{+}$, $\ket{-}$ uniformly at random. Let this state be the qubit $Q$.
\item $V_0$ sends qubit $Q$ together with $x$ to $P$. $V_1$ sends $y$ such that it arrives at $z$ at the same time as $x$ and $Q$ sent by $V_0$.
\item $P$ sends the qubit $Q$ on to $V_{f(x,y)}$.
\item If $Q$ was $\ket{0}$ or $\ket{1}$ at step $(1)$, $V_{f(x,y)}$ measures $Q$ in the computational basis. Otherwise, $V_{f(x,y)}$ measures $Q$ in the Hadamard basis. $V_0$ and $V_1$ accept if the qubit arrives at the correct time at the correct verifier and if the measurement returns the outcome consistent with the state of $Q$ at step $(1)$.
\end{enumerate}
The following proposition implies that \PVroutemod and \PVroute are essentially equivalent.
\begin{prop} \label{protocols-equivalent}
Let $p \in [0,1]$. If the attackers are caught with probability at least $p$ during \PVroutemod, then the attackers are caught with probability at least $\frac{1}{2}p$ during \PVroute. Conversely, if the attackers are caught with probability at least $p$ during \PVroute, they are caught with probability at least $p$ during \PVroutemod.
  \end{prop}
\begin{proof}
We begin by noting that preparing  $\ket{0}$, $\ket{1}$, $\ket{+}$, $\ket{-}$ with equal probability is equivalent to preparing $\ket{\Omega}$ and measuring one qubit $R$ with probability $1/2$ in the computational basis and with  probability $1/2$ in the Hadamard basis. Moreover, the other qubit $Q$ is measured in the same basis in step $(4)$ and any action on $R$ during \PVroutemod commutes with all the operations the honest prover or the attackers can do. Thus,  \PVroute is equivalent to \PVroutemod except for the final measurement, which is M2 instead of $\{\dyad{\Omega}, I_4 - \dyad{\Omega}\}$. Hence, the assertions follow from Proposition \ref{prop:measurements-equivalent}.
\end{proof}

In particular, \PVroute is secure for an adequate function $f$ since \PVroutemod is. Indeed, the following corollary is an immediate consequence of Propositions \ref{prop:repitition} and \ref{protocols-equivalent}:

\begin{cor} \label{cor:M2-repetitions}
Let $f:\{0,1\}^{2n} \to \{0,1\}$, $n \geq 10$,  $r$, $q$, $n \in \mathbb N$ and let 
\begin{equation*}
q \leq \frac{1}{2} n- 5.
\end{equation*}
Let us assume that the verifiers choose a function $f$ uniformly at random at the beginning and that they run the protocol \PVroute sequentially $r$-times, choosing $x$, $y$ uniformly at random each time. Moreover, let us assume that Alice and Bob control at most $q$ qubits each at the beginning of each iteration of \PVroute. Then, the attackers are caught with probability at least $1 - 0.98^r$.
\end{cor}

\section{The measuring protocol} \label{sec:measure}

\subsection{The measuring protocol}
In this section, we consider the measuring protocol. It turns out that we can prove similar security guarantees as for the qubit routing protocol, using essentially the same proof techniques. For the sake of analysis, we consider again a modified protocol in which $V_0$ sends half of an EPR pair, and measures the other half in the correct basis at the end of the protocol. In this case, the modified protocol is completely equivalent to the original and we will refer to both as \PVBB. For a Boolean function $f$ on $2n$ classical bits, $n \in \mathbb N$, the modified protocol is defined as follows:
\begin{enumerate}
\item $V_0$ randomly chooses two $n$-bit strings $x$, $y$, computes $f(x,y)$ and sends $y$ on to $V_1$. Moreover, $V_0$ prepares a maximally entangled $2$-qubit state $\ket{\Omega}=\frac{1}{\sqrt{2}}(\ket{00} + \ket{11})$.
\item $V_0$ sends one qubit $Q$ of $\ket{\Omega}$ together with $x$ to $P$. $V_1$ sends $y$ such that it arrives at $z$ at the same time as $x$ and $Q$ sent by $V_0$.
\item $P$ measures $Q$ in the computational basis if $f(x,y)=0$ and in the Hadamard basis if $f(x,y)=1$. Subsequently, $P$ broadcasts the measurement outcome $b \in \{0,1\}$ to both $V_0$ and $V_1$.
\item $V_0$ and $V_1$ accept if the classical bit $b$ arrives at the correct time and if a measurement on their qubit (in the computational basis if $f(x,y)=0$ and in the Hadamard basis if $f(x,y)=1$) yields the outcome $b$. 
\end{enumerate}
The timing of the response from $P$ is deemed correct if it is compatible with bit $b$ originating from $z$ right after the qubit $Q$ reached that point. An illustration of the protocol can be found in Figure 3 in the main text.

Now, we define attack strategies for the protocol.
\begin{defi}[$q$-qubit strategy for \PVBB{}] \label{defi:strategy-BB84}
Fix a partition into systems $R A \tilde A A_c B \tilde B B_c$. Both Alice's and Bob's registers each consist of $q$ qubits.
Let $d$ be the combined dimension of this system, therefore $d = 2^{2q+1}$.
A \emph{$q$-qubit strategy for \PVBB{}} consists of the starting state $\ket{\psi}$ on $R A \tilde A A_c B \tilde B B_c$,  unitaries $U^x_{A \tilde A A_c}$, $V^y_{B \tilde B B_c}$, and Alice's and Bob's local two-outcome POVMs \{$\Pi^{xy}_{A \tilde A B_c}$, $I-\Pi^{xy}_{A \tilde A B_c}$\} and  $\{\Sigma^{xy}_{B \tilde B A_c},I-\Sigma^{xy}_{B \tilde B A_c}\}$,  for all $x$, $y \in \{0,1\}^n$.
The superscripts indicate whether the operators may depend on the message $x$ that $V_0$ sends, the message $y$ that $V_1$ sends, or on both messages.
\end{defi}
See Figure 4 in the main text for an illustration. We interpret the strategy as follows: First Alice applies $U$ as a function of $x$ and Bob applies $V$ as function of $y$. Then, Alice and Bob exchange registers $A_c$ and $B_c$. Finally, Alice (with full knowledge of both $x$ and $y$) measures her local registers using a POVM given by $\{\Pi, I-\Pi\}$, responding to $V_0$ with her outcome. Similarly, Bob measures his local registers using $\{\Sigma, I-\Sigma\}$ to determine his response to $V_1$. Any unitary on Alice's or Bob's side after the communication phase, which may depend on $x$, $y$, can be absorbed into the POVMs. The same holds for any classical post-processing. The definition of an $(\epsilon, l)$-perfect $q$-qubit strategy is the same as for \PVroutemod:
\begin{defi}[$(\epsilon, l)$-perfect $q$-qubit strategy for \PVBB]
Let $\epsilon > 0$, $l \in \mathbb N$. A  $q$-qubit strategy for \PVBB as in Definition~\ref{defi:strategy} is \emph{$(\epsilon, l)$-perfect} if on $l$ pairs of strings $(x,y)$, Alice and Bob are caught by the verifiers with probability at most $\epsilon^2$.  
\end{defi}

\subsection{Lower bounds}

Our main task is to find a proposition which plays the role of Proposition \ref{prop:states-must-have-low-fidelity}. We will use entropic uncertainty relations to achieve this task.

Buhrman et al.~\cite{buhrman2014position} used an entropic uncertainty principle called the \emph{strong complementary information tradeoff} (CIT) from \cite{RB09, BCC+10} to bound the attack probability on the basic BB84 quantum PV scheme against unentangled attackers. The following version is also used in \cite[Theorem 2.4]{buhrman2014position}, where we have relabeled registers and instantiated with $n=1$.
\begin{thm}[CIT]\label{thm:CIT}
Let  $\ket{\psi_{REF}} \in \mathcal{H}_R \otimes \mathcal{H}_E \otimes \mathcal{H}_F$ be an arbitrary tri-partite state, where
$\mathcal{H}_R = \mathbb{C}^{2}$. Let the hybrid state $\rho_{ZEF}$ be obtained by measuring $R$ in basis $\theta \in \{0,1\}$, and let the hybrid state $\sigma_{ZEF}$ be obtained by measuring $R$ (of the original state $\ket{\psi_{REF}}$) in the complementary basis $\bar{\theta}$. Then, using conditional quantum entropy,
\[
H(\rho_{ZE} | E) + H(\sigma_{ZF} | F) \geq 1.
\]
\end{thm}

We start by defining sets of states from which, if they arise after the communication phase, the attackers can successfully attack the protocol.

\begin{defi} \label{def:good-sets-BB84}
Let $ \epsilon \in [0,1]$. We define $\mathcal S_0^{\epsilon, \mathrm{meas}}$ as the set of states $\ket{\varphi}_{RA\tilde A A_c B \tilde B B_c}$ such that there exists a measurement on $A \tilde A B_c$ and a measurement on $B \tilde B A_c$ which each allow to guess the outcome of a measurement performed on $R$ in the computational basis with probability at least $1-\epsilon^2$. Moreover, we define $\mathcal S_1^{\epsilon, \mathrm{meas}}$ as the set of states $\ket{\varphi^\prime}_{RA\tilde A A_c B \tilde B B_c}$ such that there exists a measurement on $A \tilde A B_c$ and a measurement on $B \tilde B A_c$ which each allow to guess the outcome of a measurement performed on $R$ in the Hadamard basis with probability at least $1-\epsilon^2$.
\end{defi}

Now, note that having a successful attack on \PVBB{} for some $x,y$ implies that the corresponding entropy is low:

\begin{lem} \label{lem:how-to-win-BB84}
Let $\epsilon \in [0,1]$ and let $\delta = h(\epsilon^2)$. Let $\ket{\varphi_0}_{RA \tilde A A_c B \tilde B B_c} \in \mathcal S_0^{\epsilon, \mathrm{meas}}$ and $\ket{\varphi_1}_{RA \tilde A A_c B \tilde B B_c} \in \mathcal S_1^{\epsilon, \mathrm{meas}}$.
Moreover, let $\rho_{Z A \tilde A A_c B \tilde B B_c}$ be the state that results after measuring register $R$ of $\ket{\phi_0}$ in the computational basis and $\sigma_{Z A \tilde A A_c B \tilde B B_c}$ the state that results after measuring register $R$ of $\ket{\phi_1}$ in the Hadamard basis. Then $H(\rho_{ZA \tilde A B_c} | A \tilde A B_c ) \leq \delta$ and $H(\rho_{Z B \tilde B A_c} | B \tilde B A_c ) \leq \delta$.
Likewise, we find that $H(\sigma_{ZA \tilde A B_c} | A \tilde A B_c ) \leq \delta$ and $H(\sigma_{Z B \tilde B A_c} | B \tilde B A_c ) \leq \delta$.
\end{lem}

\begin{proof}
First consider $\ket{\varphi_0}$ and Alice's registers $A \tilde A B_c$. Abusing notation slightly, we denote by $Z$ the random variable obtained by measuring register $R$ of $\ket{\varphi_0}$ in the computational basis, thus transforming $\ket{\varphi_0}$ into $\rho_{ZA \tilde A A_c B \tilde B B_c}$. Let $W$ be the random variable denoting Alice's outcome of the POVM measurement on local registers $A \tilde A B_c$ which allows to guess $Z$. This measurement is guaranteed to exist from the definition of $\mathcal S_0^{\epsilon, \mathrm{meas}}$ in Definition~\ref{def:good-sets-BB84}. It transforms $\rho_{ZA \tilde A A_c B \tilde B B_c}$ into $\rho_{ZWA_c B \tilde B}$.
For a probability of error $\mathbb{P}(Z \neq W) \leq \epsilon^2$, by Fano's inequality it holds that $H(Z | W) \leq h(\epsilon^2)$.
Since we have that $H(\rho_{ZA \tilde A B_c} | A \tilde A B_c ) \leq H(Z | W)$ by the data processing inequality for the relative entropy, applied to the mutual information, the statement follows directly. The other three cases can be shown analogously. 
\end{proof}

To proceed, we need to recall the continuity of conditional quantum entropy:
\begin{prop}\label{prop:entropycont}
Let $R$ be such that $\dim R = 2$, and let the dimensions of the systems $E$, $F$ be arbitrary. If $\mathcal P(\rho_{REF}, \sigma_{REF}) \leq 0.013$,  then $|H(\rho_{RE} | E) - H(\sigma_{RE} | E)| \leq 0.127$.
\end{prop}
\begin{proof}
Let $\Delta = 0.013$. The purified distance is an upper bound on the trace distance $\frac{1}{2}\norm{\cdot}_1$, see e.g.\  \cite[Lemma 3.17]{Tomamichel2016}. Thus,  $\frac{1}{2}\norm{\rho_{RE} - \sigma_{RE}}_1 \leq \Delta$, where we have used data-processing for the trace distance. The assertion follows then from the Alicki-Fannes-Winter inequality \cite[Lemma 2]{Winter2016}, which yields
\begin{equation*}
    |H(\rho_{RE} | E) - H(\sigma_{RE} | E)| \leq 2 \Delta + (1+\Delta) \, h\left(\frac{1}{1+\Delta}\right).
\end{equation*}
The assertion follows inserting the numerical value for $\Delta$.
\end{proof}

To show security, we can follow a similar strategy as for the entangled routing protocol. A key result is the following proposition:

\begin{lem} \label{lem:states-must-have-low-fidelity-BB84}
Let $\delta \leq h[(0.3)^2]$. Moreover, let $\ket{\varphi_0}_{RA\tilde{A}A_cB\tilde{B}B_c}$ be such that $H(\rho^0_{ZA\tilde{A}B_c}|A\tilde{A}B_c) \leq \delta$, where $\rho^0_{ZA\tilde{A}A_cB\tilde{B}B_c}$ is the state resulting from measuring the $R$ register of $\ket{\varphi_0}$ in the computational basis.
Similarly, let $\ket{\varphi_1}_{RA\tilde{A}A_cB\tilde{B}B_c}$ be such that $H(\sigma^1_{ZB\tilde{B}A_c}|B\tilde{B}A_c) \leq \delta$, where $\sigma^1_{ZA\tilde{A}A_cB\tilde{B}A_c}$ is the state resulting from measuring the $R$ register of $\ket{\varphi_1}$ in the Hadamard basis. Then,
\begin{equation*}
\mathcal P(\ket{\varphi_0},\ket{\varphi_1}) > 0.013.
\end{equation*}
\end{lem}

\begin{proof}
Define $\sigma^0_{ZA\tilde{A}A_cB\tilde{B}B_c}$ analogously to $\rho^0_{ZA\tilde{A}A_cB\tilde{B}B_c}$, except that the $R$ register of $\ket{\varphi_0}$ is measured in the Hadamard basis instead of the computational basis.
Then, we can fill in the CIT statement Theorem \ref{thm:CIT} to obtain the inequality
we combine with our assumption of $H(\rho^0_{ZA\tilde{A}B_c} | A\tilde{A}B_c) \leq \delta$ to get 
\[
H(\sigma^0_{ZB\tilde{B}A_c} | B\tilde{B}A_c) \geq 1 - \delta \,.
\]
Recall now that we assumed $H(\sigma^1_{ZB\tilde{B}A_c} | B\tilde{B}A_c) \leq \delta$ and that therefore
\[
|H(\sigma^0_{ZB\tilde{B}A_c} | B\tilde{B}A_c) - H(\sigma^1_{ZB\tilde{B}A_c} | B\tilde{B}A_c) | \geq 1 - 2\delta > 0.127\,.
\]
By the contrapositive of Proposition~\ref{prop:entropycont}, this implies 
\[
\mathcal P(\sigma^0_{ZB\tilde{B}A_c}, \sigma^1_{ZB\tilde{B}A_c}) > 0.013  \,.
\]
Recall $\sigma^0_{ZB\tilde{B}A_c}$ was obtained from $\ket{\phi_0}$ by tracing out $A\tilde{A}B_c$ and measuring $R$ in the Hadamard basis. Similarly, $\sigma^1_{ZB\tilde{B}A_c}$ was obtained by applying precisely the same operation to $\ket{\phi_1}$.
Therefore, the lower bound on the distance  $\mathcal P(\sigma^0_{ZB\tilde{B}A_c}, \sigma^1_{ZB\tilde{B}A_c})$ implies the same lower bound for $\mathcal P(\ket{\phi_0},\ket{\phi_1})$. This follows from data-processing for the fidelity (e.g.\ \cite[Proposition 3.2]{Tomamichel2016}).
\end{proof}

Now we are ready to state our replacement for Proposition \ref{prop:states-must-have-low-fidelity}.

\begin{prop} \label{prop:states-far-apart}
Let $0 \leq \epsilon \leq 0.3$ and let $\ket{\varphi_0}_{RA\tilde A A_c B \tilde B B_c}\in \mathcal S_0^{\epsilon, \mathrm{meas}}$, $\ket{\varphi_1}_{RA\tilde A A_c B \tilde B B_c}\in \mathcal S_1^{\epsilon, \mathrm{meas}}$. Then,
\begin{equation*}
    \mathcal P(\ket{\varphi_0}, \ket{\varphi_1}) > 0.013.
\end{equation*}
\end{prop}
\begin{proof}
This follows from combining Lemma \ref{lem:how-to-win-BB84} and Lemma \ref{lem:states-must-have-low-fidelity-BB84}.
\end{proof}

We can now proceed to proving security of the measuring protocol.
\begin{prop}\label{prop:no-perfect-strategies} Let $0 \leq \epsilon \leq 0.3$ and $n \geq 10$, $q$, $n \in \mathbb N$. Then, a uniformly random function $f:\{0,1\}^{2n} \to \{0,1\}$ has the following property with probability at least $1-2^{-2^{n}}$: 
Any $(\epsilon,\frac{3}{4}\cdot 2^{2n})$-perfect $q$-qubit strategy for \PVBB requires
\begin{equation*}
q > \frac{1}{2} n - 5,
\end{equation*}
where $\ket{\psi}$ is a state on $2q+1$ qubits.
\end{prop}
\begin{proof}
The statement follows from Proposition \ref{prop:states-far-apart}, Lemma \ref{lem:classical_compression} and Lemma \ref{lem:stable_E4} in the same way as in the proof of Proposition \ref{prop:after-embezzlement}.
\end{proof}

\begin{thm} \label{thm:second-main-result}
Let $f:\{0,1\}^{2n} \to \{0,1\}$, $n \geq 10$, $r$, $q$, $n \in \mathbb N$ and let 
\begin{equation*}
q \leq \frac{1}{2} n- 5.
\end{equation*}
Let us assume that the verifiers choose a function $f$ uniformly at random at the beginning and that they run the protocol \PVBB sequentially $r$-times, choosing $x$, $y$ uniformly at random each time. Moreover, let us assume that Alice and Bob control at most $q$ qubits each at the beginning of each iteration of \PVBB. Then, the attackers are caught with probability at least $1 - 0.98^r$.
\end{thm}
\begin{proof}
The assertion follows from Proposition \ref{prop:no-perfect-strategies} along the lines of the proofs of Theorem \ref{thm:main-result} and Proposition \ref{prop:repitition}.
\end{proof}

\begin{remark}
Note that \PVBB would still be secure if we replaced the requirement that the honest prover needs to send the bit $b$ to both verifiers at the end of the protocol by requiring that $P$ sends $b$ only to $V_{f(x,y)}$. This can be seen from the proof of Lemma \ref{lem:states-must-have-low-fidelity-BB84}.
\end{remark}

\section{Resistance to noise} \label{sec:noise}

Finally, we consider the effect of noise on \PVroute and \PVBB.  Let us now assume that the noise in the experiment causes the honest prover 
to be rejected with probability at most $\eta$. In order to deal with the noise, the verifiers will repeat the protocol independently $r$-times and accept if the individual rounds accept more than $0.996(1-\eta)r$ times. We will call such protocols \emph{${PV}_{\mathrm{noisy, \#}}^f(r)$ with noise level $\eta$}, where $\# \in \{\mathrm{route}, \mathrm{meas}\}$. The next theorem shows that such protocols are still secure.

\begin{thm} \label{thm:noise-robustness}
Let  $r$, $q$, $n \in \mathbb N$, $n \geq 10$, $0 \leq \eta \leq 10^{-2}$. Assume that a function $f:\{0,1\}^{2n} \to \{0,1\}$ is chosen uniformly at random. Then, an honest prover  succeeds in ${PV}_{\mathrm{noisy, \#}}^f(r)$ with noise level $\eta$ with probability at least
\begin{equation*}
   1-c^r,
\end{equation*}
where $\# \in \{\mathrm{route}, \mathrm{meas}\}$. Attackers controlling at most $q \leq \frac{1}{2}n- 5$ qubits each round will succeed with probability at most
\begin{equation*}
    c'^r,
\end{equation*}
where $c, c'<1$ are universal constants. In particular, we can choose $c =c'= \exp(-8\cdot 10^{-6})$.
\end{thm}
\begin{proof}
Let $X_i$ be random variables which are $1$ if the honest prover succeeds at round $i$ and $0$ if she fails. Let $X := \sum_{i = 1}^r X_i$. Then, the probability that the honest prover succeeds at $PV_{\mathrm{noisy, \#}}^f(r)$ with noise level $\eta$ is $\mathbb P[X > 0.996(1-\eta)r]$. Since in each round the honest prover accepts with probability at least $0.99$, this implies that $\mathbb P(X_i = 1| X_{i-1} = x_{i-1}, \ldots, X_{1} = x_{1}) \geq 0.99$ for any $x_j \in \{0,1\}$, $j \in \{1, \ldots, r\}$.

Let $X^\prime_i$ be i.i.d.\ random variables which are $1$ with probability $0.99$ and $0$ with probability $10^{-2}$. Let $X^\prime = \sum_{i = 1}^r X^\prime_i$. Then, the probability that the honest prover succeeds can be bounded using the random variable $X^\prime$ as  $\mathbb P[X^\prime > 0.996(1-\eta)r]$ by Lemma \ref{lem:bound-by-iid}. We estimate
\begin{align*}
    \mathbb P[X^\prime > 0.996(1-\eta)r] &= 1- \mathbb P[X^\prime \leq 0.996(1-\eta)r] \\
    & \geq 1- e^{-\frac{r (1-\eta) 16\cdot 10^{-6}}{2}},
\end{align*}
where we have used the Chernoff bound. Inserting the bound on $\eta$, the first assertion follows.

Likewise, let $Y_i$ be a random variable which is $1$ if the attackers succeed in round $i$ and $0$ if they do not, $i \in \{1, \ldots r\}$. Let $Y = \sum_{i = 1}^r Y_i$. Then, the probability that the attackers succeed is $\mathbb P[Y > 0.996(1-\eta)r]$. Using the same argument as in Proposition \ref{prop:repitition}, Corollary \ref{cor:M2-repetitions} and Theorem \ref{thm:second-main-result}, respectively, yields that $\mathbb P(Y_i = 1| Y_{i-1} = y_{i-1}, \ldots, Y_{1} = y_{1}) \leq 0.98$ for any $y_j \in \{0,1\}$, $j \in \{1, \ldots, r\}$.

Moreover, let $Y^\prime_i$ be i.i.d.\ random variables which are $1$ with probability $0.98$ and $0$ with probability $2 \cdot 10^{-2}$. Additionally, let $Y^\prime = \sum_{i = 1}^r Y_i^\prime$. Then, by Lemma \ref{lem:bound-by-iid}, the probability that the attackers succeed is at most $\mathbb P[Y^\prime > 0.996(1-\eta)r]$. Let us define $\eta^\prime = 2 \cdot 10^{-2}$. Solving the equation 
\begin{equation*}
    0.996 (1-\eta) = (1+\delta^\prime)(1-\eta^\prime)
\end{equation*}
for $\delta^\prime$, we obtain $\delta^\prime \geq 0.99\cdot\frac{996}{980}-1> 0$. We make a similar estimate as before, 
\begin{align*}
    \mathbb P[Y^\prime > 0.996(1-\eta)r] &= \mathbb P[Y^\prime > (1+\delta^\prime)(1-\eta^\prime)r] \\
    &\leq e^{-\frac{r (1-\eta^\prime) (\delta^\prime)^2}{3}},
\end{align*}
where we have used the Chernoff bound again. Inserting the expressions for $\eta^\prime$ and bounding $\delta^\prime \geq 5\cdot 10^{-3}$, the second assertion follows.
\end{proof}

\section{Lower bound for concrete functions} \label{sec:concrete-functions}
\label{sec:lbconcrete}
In this section, we will finally consider concrete functions $f$ instead of uniformly random ones and prove that \PVroute and $PV^f_{\mathrm{meas}}$ are still secure against bounded attackers, although the bounds are weaker than for random functions $f$. The proofs use a connection of classical roundings to the communication complexity of $f$.

Let us fix some function $f:\{0,1\}^n \times \{0,1\}^n \to \{0,1\}$. We define $D^{1,\mu}_\epsilon(f)$ as the \emph{one-way distributional communication complexity} of the function~$f$  under some distribution $\mu$~(see \cite[Definition 3.19]{Kushilevitz1996} for the (two-way) distributional communication complexity $D^{\mu}_\epsilon(f)$).
This represents the amount of (classical) bits Alice needs to send for Bob in a deterministic protocol, for Bob to be able to compute $f(x,y)$ correctly with probability $1-\epsilon$, where the probability is taken over $(x,y)$ pairs drawn from the input distribution $\mu$.

Similarly, let $D^{\parallel,\mu}_\epsilon(f)$ be the distributional communication complexity of a function $f$ in the simultaneous message passing (SMP) model.
Here Alice and Bob both are allowed to send a single message to a third party, the referee, who has to output the function value given these messages. We take the distributional communication complexity in the SMP model as the length of the longest message, not the sum of the length of both messages.
Several lower bounds for this model are given for  $D^{1,\mu}_\epsilon(f)$, but it's easy to see that  $D^{\parallel,\mu}_\epsilon(f) \geq D^{1,\mu}_\epsilon(f)$.

In the other lower bounds of this work, we have restricted our analysis to the case of a uniform distribution over the input pairs. The following analysis holds for any input distribution, but for simplicity we will only consider the uniform distribution. Let $u$ denote the uniform distribution over all pairs of $n$-bit strings (where $n$ will be clear from context).
For example, for the inner product function \eqref{eq:ip}, we have that $D^{1,u}_{1/2-\epsilon}(IP) \geq n/2 - \log(1/\epsilon)-1$~\cite[Example 3.29]{Kushilevitz1996}, since $D^{\mu}_\epsilon(f) \leq D^{1,\mu}_\epsilon(f)+1$.

By using the classical roundings developed for \PVroutemod and $PV^f_{\mathrm{meas}}$, we can show that for a wide range of explicit functions,
the attackers need to manipulate a number of qubits that is logarithmic in the number of bits $n$.
This bound is exponentially worse than the one we obtain for \emph{random} functions $f$, but already holds for explicitly-defined easily-computable functions, such as the inner-product function\footnote{Recall that for the example of the inner-product function, it is not hard to construct an attack that uses $n$ EPR pairs \cite{Buhrman2013}. }.
So, for an explicit easily-computable function, the ratio of entanglement that attackers need also grows unboundedly with the classical information involved (but with a worse bound on the dependence of the number of classical bits $n$ than we obtained for a random function), while the honest parties only need to manipulate a single qubit.

This can be viewed as a robust version of \cite[Theorem E.3]{Buhrman2013}, which showed that perfect attacks on any injective function (which were effectively functions with maximal deterministic one-way communication complexity) need at least $\Omega(\log(n))$ qubits.

\begin{prop}\label{prop:CC-lower-bounds}
Let $\epsilon \leq \epsilon_0$,  where $\epsilon_0$ is chosen according to the requirements of Lemma~\ref{lem:classical_compression}. Moreover, let $f$ be such that $D^{\parallel,u}_{1/4}(f) \geq k$, where $u$ is the uniform distribution.
Then there exists no $(\epsilon, \frac{3}{4}\cdot2^{2n})$-perfect $q$-qubit strategy for either \PVBB or \PVroutemod, with \[
\log(927)2^{2q+2} < k \,,
\]
implying no such strategy exists for
\[
q \leq \frac{1}{2}\log k - 3 \,.
\]
\end{prop}
\begin{proof}
We prove the statement by contradiction: Any assumed strategy on a low number of qubits can be directly converted into a classical communication protocol for solving $f$.
The only required observation is that the classical compression of $f$ that we get as a result of Lemma~\ref{lem:classical_compression} not only encodes a full description of the function $f$, but its parts can also be evaluated on specific $x$ and $y$ to get a communication protocol for $f$ in the required simple form.

Assume, for a contradiction, that  a $(\epsilon, \frac{3}{4} \cdot 2^{2n})$-perfect $q$-qubit strategy exists. Then,  Lemma~\ref{lem:classical_compression} implies the existence of an $(\epsilon,q)$-classical rounding of $f$ of size $k=\log(927)2^{2q+2}$.
Therefore, by Definition~\ref{def:classical-rounding}, there exists a function $g :\{0,1\}^{3k} \to \{0,1\}$, functions $f_A: \{0,1\}^{n} \to \{0,1\}^k$, $f_B: \{0,1\}^{n} \to \{0,1\}^k$, and a constant $\lambda \in \{0,1\}^k$ such that for at least $\frac{3}{4}$ of the input pairs $(x,y)$ it holds that
\[
f(x,y) = g(f_A(x),f_B(y),\lambda) \,.
\]
Given that the function $f$ and the strategy are known beforehand, all parties can precompute these objects in the communication complexity setting.

The simultaneous-message passing protocol now simply proceeds as follows: Alice sends the $k$-bit string $s=f_A(x)$ to the referee, and Bob sends the $k$-bit string $t=f_B(y)$. The referee computes $g(s,t,\lambda)$ and outputs this as function value.
\end{proof}

\begin{thm} \label{thm:comm-complex}
Let $f$ be such that $D^{\parallel,u}_{1/4}(f) \geq k$, where $u$ is the uniform distribution, and let
\begin{equation*}
    q \leq \frac{1}{2}\log k - 3.
\end{equation*}
If $x$, $y$ are chosen uniformly at random during the protocols, then attackers controlling at most $q$ qubits each are detected  during \PVroute and $PV_{\mathrm{meas}}^{f}$ with probabilities at least $2 \cdot 10^{-2}$, respectively.
\end{thm}
\begin{proof}
Using Proposition \ref{prop:CC-lower-bounds}, this follows from Propositions \ref{prop:states-must-have-low-fidelity} and \ref{prop:states-far-apart} in a similar way as Theorem \ref{thm:main-result}. 
\end{proof}
For the inner product function, this implies Theorem \ref{thm:intro-IP-lower-bounds}.

\begin{remark}
Replacing the upper bound on $q$ by $q \leq \frac{1}{2}\log k - 3$ for $f$ be such that $D^{\parallel,u}_{1/4}(f) \geq k$ in Proposition \ref{prop:repitition}, Theorem \ref{thm:second-main-result} and Theorem \ref{thm:noise-robustness}, we can derive the corresponding statements on repetition and noise robustness for a concrete function by following the exact same proof strategies. Moreover, we could consider \PVroute instead of \PVroutemod, decreasing the detection probability by a factor $1/2$.
\end{remark}

\section{Attack model and comparison to previous work} \label{sec:attack-model}
When analyzing protocols for quantum PV in a resource-bounded setting, care has to be taken with respect to what is counted exactly.
A fair comparison will involve weighing the resources required of an honest party compared to those of the attackers, proving hopefully that any attack is much harder to perform than executing the protocol.
Which resources are important, and how to weigh them, is not trivial and there are several choices to be made in how to count the resources involved.
These choices include:
\begin{itemize}
    \item Do we only count quantum information manipulated by the attackers and the honest parties, or do we also quantify classical information?
    \item Do we look at the size of \emph{all} quantum resources required, or do we just want to limit the pre-shared state of the attackers?
    \item Do we allow quantum communication between the attackers, or do we assume this communication to be classical and subsume these messages in the entanglement by way of teleportation?
    \item Would it be possible to bound the resources using something else than the number of qubits, such as entanglement entropy?
\end{itemize}
How the strength of attack resource lower bounds should be interpreted, depends on which choices are made here.

In this work, we only count the quantum information: $x$ and $y$ are distributed amongst the attackers for free\footnote{I.e., the first step of the attack lets Alice see $x$ and Bob see $y$, and after the messages are exchanged both attackers know $(x,y)$.}, and we bound the total number of qubits each of the attackers utilize.

Counting in such a way, Theorems \ref{thm:intro} and \ref{thm:intro-2} imply that the amount of \emph{quantum} resources used by the attackers is \emph{unbounded}\footnote{This comparison is at most exponential if all information is counted, via the attack of Beigi and K\"onig~\cite{Beigi2011}, but has no a-priori bound if classical communication is considered free.} as a function of the quantum information manipulated by the honest party. Indeed, our bounds show that the number of qubits manipulated by the attackers grows linearly as the amount of classical information grows while the honest party only manipulates a single qubit.

Note that there still is a gap between the best known attack (needing $2^n$ EPR pairs, for an honest protocol with $n$ classical bits and one single qubit) and our lower bounds, when we look at how the requirement itself grows as function of $n$.
This gap is not evident when only looking at how the respective quantum requirements relate as a function of each other.

The choices made in the attack model also influence the comparison to other results. The independent recent work by  Junge, Kubicki, Palazuelos, and P\'erez-Garc\'ia~\cite{Junge2021} uses an attack model which is very close to ours. The authors do not count classical communication either, but only compare the quantum resources needed by the honest prover compared to the attackers. In the case that the honest prover has to manipulate quantum systems with $2\log{n}$ qubits, the authors can show that the attackers need a quantum system of $\Omega(n^\alpha)$ qubits for some $\alpha > 0$, provided that the attack being used is \emph{smooth}. The classical information that the honest prover has to manipulate during this protocol is $n^2$ bits. The smoothness requirement covers all known attacks. Furthermore, the authors put forward a conjecture in Banach space theory that, if true, would allow to remove the smoothness assumption, and give evidence for it. 

When only comparing quantum resources, our bounds are stronger in the sense that while the ratio of quantum resources in  \cite{Junge2021} is exponential, the ratio in the qubit routing and measuring protocols is unbounded. The trade-off between the classical information sent and the number of qubits needed by the attackers is similar in all cases. On the other hand, \cite{Junge2021} establishes the link to geometric functional analysis, which could allow to tackle the ultimate goal, i.e., showing that the quantum resources the attackers need are exponential in \emph{all} the resources the honest party needs.

When it comes to previous protocols for quantum PV, the best studied is the BB84-type protocol \cite{Kent2011,Lau2011, buhrman2014position,Beigi2011, tomamichel2013monogamy, Ribeiro2015}, which was the inspiration for our measuring protocol. All the bounds for this protocol are linear in the sense that the honest prover needs to manipulate $n$ qubits while the attackers need $\Omega(n)$ qubits to break the protocol. The improvement of \cite{tomamichel2013monogamy} over \cite{Beigi2011} is that only single qubit measurements are necessary. Our routing protocol is simpler in the sense that the honest party only needs to route one qubit instead of measuring $n$ qubits separately. On the other hand, the BB84-type protocol needs the honest prover to send back merely classical information, whereas the qubit routing protocol requires to send back quantum information. The measuring protocol remedies this fact while the honest prover still needs to manipulate a single qubit.
However, the detection probability of attackers can be made arbitrarily large using parallel repetition of the BB84-type protocol as shown in \cite{tomamichel2013monogamy}, while we can increase the probability in both our protocols only through sequential repetition.

An essential difference with respect to the proof technique is that \cite{Beigi2011} showed that bounds on the success probability of the attackers without entanglement can be lifted to bound the success probability with pre-shared entanglement. In our case, this technique no longer works and we have to resort to different methods. 

Finally, \cite{Gonzales2019} also proves linear lower bounds for protocols based on non-local quantum computation (the BB84-type protocol is of the same type). However, the authors bound the entanglement entropy of the attackers instead of the dimension of their quantum systems. As a downside, their attack model does not allow for quantum communication of the attackers. For the BB84-type protocol, the latter restriction was removed in \cite{tomamichel2013monogamy} compared to \cite{Beigi2011} (but again present in the assumptions of \cite{Ribeiro2015}).

\section{Technical results} \label{sec:technical}
First, we restate and prove \cite[Lemma E.1]{Buhrman2013} in order to make the main argument self-contained.
\begin{lem}\label{lem:E1}
Let $\ket{\psi_0}$, $\ket{\psi_1}$ be states on $RA\tilde A A_c B \tilde B B_c$ and such that there are unitaries $K_{A \tilde A B_c}$, $L_{B \tilde B A_c}$ and states $\ket{\varphi_0}_{\tilde A A_c B \tilde B B_c}$, $\ket{\varphi_1}_{A \tilde A A_c \tilde B B_c}$ which satisfy
\begin{align*}
    K_{A \tilde A B_c} \ket{\psi_0}_{RA\tilde A A_c B \tilde B B_c} &= \ket{\Omega}_{RA} \otimes \ket{\varphi_0}_{\tilde A A_c B \tilde B B_c} \\
    L_{B \tilde B A_c} \ket{\psi_1}_{RA\tilde A A_c B \tilde B B_c} &= \ket{\Omega}_{RB} \otimes \ket{\varphi_1}_{A \tilde A A_c \tilde B B_c}.
\end{align*}
Then, 
\begin{equation*}
    |\langle \psi_0 | \psi_1 \rangle| \leq \frac{1}{2}.
\end{equation*}
\end{lem}
\begin{proof}
Note that $K$ and $L$ commute. We find that 
\begin{align*}
    |\langle \psi_0 | \psi_1 \rangle| &=  |\bra{\Omega}_{RA} \otimes \bra{\varphi_0} L^\ast K \ket{\Omega}_{RB} \otimes \ket{\varphi_1}| \\
    &= |\bra{\Omega}_{RA} \otimes \bra{\varphi_0^\prime}\ket{\Omega}_{RB} \otimes \ket{\varphi_1^\prime}|,
\end{align*}
where $\ket{\varphi_0^\prime}_{\tilde A A_c B \tilde B B_c} = L_{B \tilde B A_c}\ket{\varphi_0}_{\tilde A A_c B \tilde B B_c}$ and $\ket{\varphi_1^\prime}_{A \tilde A A_c \tilde B B_c} = K_{A \tilde A B_c} \ket{\varphi_1}_{A \tilde A A_c \tilde B B_c}$.
Now, 
\begin{equation*}
    \bra{\Omega}_{RA}\ket{\Omega}_{RB} = \frac{1}{2} (\ket{0}_B \bra{0}_A + \ket{1}_B \bra{1}_A)
\end{equation*}
Note that $\ket{0}_B \bra{0}_A + \ket{1}_B \bra{1}_A$ is a unitary operator from $A$ to $B$, which transfers a qubit from $A$ to $B$. Writing
\begin{equation*}
    \ket{\varphi_1^{\prime \prime}}_{\tilde A A_c B\tilde B B_c}= (\ket{0}_B \bra{0}_A + \ket{1}_B \bra{1}_A)  \ket{\varphi_1^\prime}_{A \tilde A A_c \tilde B B_c},
\end{equation*}
we infer
\begin{align*}
     |\langle \psi_0 | \psi_1 \rangle| &= \frac{1}{2} |\langle \varphi_0^\prime | \varphi_1^{\prime \prime} \rangle|\\
     &\leq \frac{1}{2},
\end{align*}
since both $\ket{\varphi_0^\prime}$ and $\ket{\varphi_1^{\prime \prime}}$ are states on $\tilde A A_c B \tilde B B_c$.
\end{proof}

We now show that the measurement of the verifiers at the end of \PVroutemod can be replaced by a measurement implemented via local measurements and classical post-processing which performs almost as good. This implies in particular that the verifiers need not store their qubit until the end of the protocol but can measure it right away. This fact is well-known in the context of the BB84 protocol \cite{Bennett1984}. We give a proof here for convenience.

We compare the two measurement procedures for some state $\rho$ on $\mathbb C^2 \otimes \mathbb C^2$.
\begin{itemize}
    \item M1: Measure $\{\dyad{\Omega}, I_4-\dyad{\Omega}\}$
    \item M2: With probability $\frac{1}{2}$ each, either measure each qubit in the computational or Hadamard basis and check whether the measurement outcomes are equal. In other words, measure either $\{\dyad{++} + \dyad{--}, I_4-(\dyad{++} + \dyad{--})\}$ or $\{\dyad{00} + \dyad{11}, I_4 - (\dyad{00} + \dyad{11})\}$, where the choice is uniformly random.
\end{itemize}
The two measurements are equivalent in the following sense:

\begin{prop} \label{prop:measurements-equivalent}
Let $\rho$ be a quantum state on two qubits and let $\delta > 0$.
\begin{enumerate}
    \item If M1 accepts with probability at least $1-\delta$, then M2 accepts with probability at least $1-\delta$.
    \item If M2 accepts with probability at least $1-\delta$, then M1 accepts with probability at least $1-2\delta$.
\end{enumerate}
\end{prop}
\begin{proof} Let 
\begin{equation*}
    \ket{\phi_1} := \frac{1}{\sqrt{2}}(\ket{01} + \ket{10}) \qquad \mathrm{and} \qquad \ket{\phi_2} := \frac{1}{\sqrt{2}}(\ket{00} - \ket{11}).
\end{equation*}  It can be verified that 
\begin{equation*}
    \ket{\Omega} = \frac{1}{\sqrt{2}}(\ket{++} + \ket{--})  \qquad \mathrm{and} \qquad \ket{\phi_1} = \frac{1}{\sqrt{2}}(\ket{++} - \ket{--}).
\end{equation*}
Thus, $\dyad{++} + \dyad{--} = \dyad{\Omega} + \dyad{\phi_1}$. Likewise, $\dyad{00} + \dyad{11} = \dyad{\Omega} + \dyad{\phi_2}$ and $\ket{\Omega}$, $\ket{\phi_1}$, $\ket{\phi_2}$ are orthogonal.

Let $p_\Omega=\bra{\Omega}\rho\Ket{\Omega}$ and $p_i=\bra{\phi_i}\rho\Ket{\phi_i}$ for $i \in \{1, 2\}$. Then, the probability that M1 accepts is $p_\Omega$, whereas the probability that M2 accepts is $p_\Omega + \frac{1}{2}p_1 + \frac{1}{2} p_2$. Thus, the first assertion follows straightforwardly. For the second assertion, note that
\begin{equation*}
    p_\Omega + p_1 + p_2 \leq 1,
\end{equation*}
since the corresponding states are orthogonal. Thus, rearranging this inequality and combining it with the assumption that M2 accepts with probability at least $1-\delta$,
\begin{equation*}
    \frac{1}{2} + \frac{1}{2}p_{\Omega} \geq 1-\delta.
\end{equation*}
The second assertion thus follows by rearranging the inequality.
\end{proof}

We conclude with a small lemma concerning discrete-time stochastic processes.
\begin{lem}\label{lem:bound-by-iid}
Let $t \in \mathbb R$,  $r \in \mathbb N$, $Y := \sum_{i=1}^r Y_i$, where $Y_1, \ldots, Y_r$ is a discrete-time stochastic process, and $Y_i \in \{0,1\}$. Let $p \in [0,1]$ and $Y^\prime = \sum_{i=1}^r Y^\prime_i$, where $Y_i^\prime := 1$ with probability $p$ and $Y^\prime_i = 0$ with probability $1-p$. It holds that 
\begin{enumerate}
    \item If $p(Y_i = 1| Y_{i-1}= y_{i-1}, \ldots, Y_{1}= y_{1}) \leq p$ for all $y_j \in \{0,1\}$, $j \in \{1, \ldots, i-1\}$ and all $i \in \{1, \ldots, r\}$, then $ \mathbb P(Y^\prime \geq t) \geq  \mathbb P(Y \geq t)$
     \item If $p(Y_i = 1| Y_{i-1}= y_{i-1}, \ldots, Y_{1}= y_{1}) \geq p$ for all $y_j \in \{0,1\}$, $j \in \{1, \ldots, i-1\}$ and all $i \in \{1, \ldots, r\}$, then $ \mathbb P(Y^\prime \geq t) \leq  \mathbb P(Y \geq t)$
\end{enumerate}
\end{lem}

\begin{proof}
We will only show the first assertion, since the second follows in a similar manner. Let $X := \sum_{i=1}^r X_i$, where $X_1, \ldots, X_r$ is a discrete-time stochastic process and $X_i \in \{0,1\}$. Set $x_i \in \{0,1\}$, $i \in \{0, \ldots, r\}$ and $\bar x := (x_1, \ldots, x_r)$. Moreover, let $|\bar x| := x_1 + \ldots + x_r$. We write 
\begin{equation*}
    \mathbb P(x_r, \ldots x_1) := \mathbb P(X_r = x_r, \ldots, X_1 = x_1) 
\end{equation*}
and use a similar notation for conditional expectations. Fix $j \in \mathbb N$ and let us assume that $p(x_i|x_{i-1}, \ldots, x_1)=p(x_i)$ for any $i \geq j+1$ and that $p(1|x_{j-1}, \ldots, x_1)\leq p$. We claim that $\mathbb P(X \geq t) \leq \mathbb P(X^\prime \geq t) $, where $X^\prime := X^\prime_j + \sum_{i\in \{1, \ldots, r\}\setminus\{j\}} X_i$ and $X_j^\prime$ is a random variable independent of $X_1, \ldots, X_r$ such that $\mathbb P(X_j^\prime = 1) = p$, $\mathbb P(X_j^\prime = 0) = 1- p$. We have thus replaced $X_j$ in $X$ by $X^\prime_j$ to obtain $X^\prime$. The first assertion then follows from an iterated application of the claim. We now prove the claim. Let $\check X := \sum_{i=1, i \neq j}^r X_i$. Then,
\begin{align}
    \mathbb P(X^\prime \geq t) &= p\mathbb P(\check X \geq t) + (1-p)\mathbb P(\check X \geq t) + p \mathbb P(\check X = t-1) \nonumber\\
    &= \mathbb P(\check X \geq t) + p \mathbb P(\check X = t-1), \label{eq:Xprime}
\end{align}
since the order of the random variables $X_i$ for $i > j$ does not matter and we can put $X^\prime_j$ last. Likewise,
\begin{align}
   \mathbb P(X \geq t) &= \sum_{\bar x: |\bar x| \geq t} \mathbb P(x_1, \ldots, x_r) \nonumber \\
   &= \sum_{\bar x:~|\bar x|-x_j \geq t} \mathbb P(x_1, \ldots, x_r) + \sum_{\substack{\bar x:~|\bar x|-x_j = t-1,\\x_j = 1}} \mathbb P(x_1, \ldots, x_r) \label{eq:x}
\end{align}
For the first term, we compute
\begin{align*}
    &\sum_{\bar x:~|\bar x|-x_j \geq t} \mathbb P(x_1, \ldots, x_r)\\&=\sum_{\bar x:~|\bar x|-x_j \geq t} \prod_{k = j+1}^r \mathbb P(x_k)\prod_{i = 1}^j \mathbb P(x_i|x_{i-1}, \ldots, x_1) \\
    &=\sum_{\bar x\setminus \{x_j\}:~|\bar x|-x_j \geq t} \prod_{k = j+1}^r \mathbb P(x_k)[\mathbb P(1|x_{j-1}, \ldots, x_1)+ \mathbb P(0|x_{j-1}, \ldots, x_1)]\prod_{i = 1}^{j-1} \mathbb P(x_i|x_{i-1}, \ldots, x_1) \\
    &=\mathbb P(\check X \geq t),
\end{align*}
where we have used that $\mathbb P(1|x_{j-1}, \ldots, x_1)+ \mathbb P(0|x_{j-1}, \ldots, x_1) = 1$. For the second term,
\begin{align*}
    \sum_{\substack{\bar x:~|\bar x|-x_j = t-1,\\x_j = 1}} \mathbb P(x_1, \ldots, x_r) &= \sum_{\substack{\bar x:~|\bar x|-x_j = t-1,\\x_j = 1}}\prod_{k = j+1}^r \mathbb P(x_k)\mathbb P(1|x_{j-1}, \ldots, x_1)\prod_{i = 1}^{j-1} \mathbb P(x_i|x_{i-1}, \ldots, x_1) \\
    \leq p \mathbb P(\check X = t-1)
\end{align*}
Thus, inserting the expressions into \eqref{eq:x} and using \eqref{eq:Xprime},
\begin{equation*}
      \mathbb P(X \geq t) \leq \mathbb P(\check X \geq t) + p \mathbb P(\check X = t-1) = \mathbb P(X^\prime \geq t).
\end{equation*}
\end{proof}


\begin{thebibliography}{25}

\bibitem{Chandran2009A}
N.~Chandran, V.~Goyal, R.~Moriarty, and R.~Ostrovsky, ``Position based
  cryptography,'' in {\em Advances in Cryptology - CRYPTO 2009}, pp.~391--407,
  Springer, 2009.

\bibitem{buhrman2014positionA}
H.~Buhrman, N.~Chandran, S.~Fehr, R.~Gelles, V.~Goyal, R.~Ostrovsky, and
  C.~Schaffner, ``Position-based quantum cryptography: Impossibility and
  constructions,'' {\em SIAM Journal on Computing}, vol.~43, no.~1,
  pp.~150--178, 2014.

\bibitem{Kent2011A}
A.~Kent, W.~J. Munro, and T.~P. Spiller, ``Quantum tagging: Authenticating
  location via quantum information and relativistic signaling constraints,''
  {\em Physical Review A}, vol.~84, p.~012326, 2011.

\bibitem{Malaney2010A}
R.~A. Malaney, ``Location-dependent communications using quantum
  entanglement,'' {\em Physical Review A}, vol.~81, no.~4, p.~042319, 2010.

\bibitem{Vaidman2003A}
L.~Vaidman, ``Instantaneous measurement of nonlocal variables,'' {\em Physical
  Review Letters}, vol.~90, p.~010402, Jan 2003.

\bibitem{Beigi2011A}
S.~Beigi and R.~K\"onig, ``Simplified instantaneous non-local quantum
  computation with applications to position-based cryptography,'' {\em New
  Journal of Physics}, vol.~13, no.~9, p.~093036, 2011.

\bibitem{tomamichel2013monogamyA}
M.~Tomamichel, S.~Fehr, J.~Kaniewski, and S.~Wehner, ``A
  monogamy-of-entanglement game with applications to device-independent quantum
  cryptography,'' {\em New Journal of Physics}, vol.~15, no.~10, p.~103002,
  2013.

\bibitem{Ribeiro2015A}
J.~Ribeiro and F.~Grosshans, ``A tight lower bound for the {BB84}-states
  quantum-position-verification protocol,'' {\em arXiv-preprint
  arXiv:1504.07171}, 2015.

\bibitem{Lau2011A}
H.-K. Lau and H.-K. Lo, ``Insecurity of position-based quantum-cryptography
  protocols against entanglement attacks,'' {\em Physical Review A}, vol.~83,
  p.~012322, 2011.

\bibitem{Chakraborty2015A}
K.~Chakraborty and A.~Leverrier, ``Practical position-based quantum
  cryptography,'' {\em Physical Review A}, vol.~92, p.~052304, 2015.

\bibitem{Malaney2016A}
R.~Malaney, ``The quantum car,'' {\em IEEE Wireless Communications Letters},
  vol.~5, no.~6, pp.~624--627, 2016.

\bibitem{Das2017A}
S.~Das and G.~Siopsis, ``Practically secure quantum position verification,''
  {\em New Journal of Physics}, vol.~23, p.~063069, 2021.

\bibitem{Gonzales2019A}
A.~Gonzales and E.~Chitambar, ``Bounds on instantaneous nonlocal quantum
  computation,'' {\em IEEE Transactions on Information Theory}, vol.~66, no.~5,
  pp.~2951--2963, 2019.

\bibitem{Buhrman2013A}
H.~Buhrman, S.~Fehr, C.~Schaffner, and F.~Speelman, ``The garden-hose model,''
  in {\em Proceedings of the 4th Conference on Innovations in Theoretical
  Computer Science}, ITCS '13, pp.~145--158, ACM, 2013.

\bibitem{Speelman2016A}
F.~Speelman, ``Instantaneous non-local computation of low {T}-depth quantum
  circuits,'' in {\em 11th Conference on the Theory of Quantum Computation,
  Communication and Cryptography (TQC 2016)}, vol.~61 of {\em Leibniz
  International Proceedings in Informatics (LIPIcs)}, pp.~9:1--9:24, Schloss
  Dagstuhl--Leibniz-Zentrum für Informatik, 2016.

\bibitem{Olivo2020A}
A.~Olivo, U.~Chabaud, A.~Chailloux, and F.~Grosshans, ``Breaking simple quantum
  position verification protocols with little entanglement,'' {\em arXiv
  preprint arXiv:2007.15808}, 2020.

\bibitem{Kent2010A}
A.~Kent, ``Quantum tagging with cryptographically secure tags,'' {\em arXiv
  preprint arXiv:1008.5380}, 2010.

\bibitem{Gao2013A}
F.~Gao, B.~Liu, and Q.-Y. Wen, ``Enhanced no-go theorem for quantum position
  verification,'' {\em arXiv-preprint arXiv:1305.4254}, 2013.

\bibitem{Unruh2014A}
D.~Unruh, ``Quantum position verification in the random oracle model,'' in {\em
  Annual Cryptology Conference}, pp.~1--18, Springer, 2014.

\bibitem{Qi2015A}
B.~Qi and G.~Siopsis, ``Loss-tolerant position-based quantum cryptography,''
  {\em Physical Review A}, vol.~91, p.~042337, 2015.

\bibitem{Lim2016A}
C.~C.~W. Lim, F.~Xu, G.~Siopsis, E.~Chitambar, P.~G. Evans, and B.~Qi,
  ``Loss-tolerant quantum secure positioning with weak laser sources,'' {\em
  Physical Review A}, vol.~94, p.~032315, 2016.

\bibitem{Allerstorfer2021A}
R.~Allerstorfer, H.~Buhrman, F.~Speelman, and P.~{Verduyn Lunel}, ``New
  protocols and ideas for practical quantum position verification,'' {\em
  arXiv-preprint arXiv:2106.12911}, 2021.

\bibitem{Junge2021A}
M.~Junge, A.~M. Kubicki, C.~Palazuelos, and D.~{P{\'e}rez-Garc{\'i}a},
  ``Geometry of {Banach} spaces: a new route towards position based
  cryptography,'' {\em arXiv-preprint arXiv:2103.16357}, 2021.

\bibitem{RB09A}
J.~M. Renes and J.-C. Boileau, ``Conjectured strong complementary information
  tradeoff,'' {\em Physical Review Letters}, vol.~103, p.~020402, 2009.

\bibitem{BCC+10A}
M.~Berta, M.~Christandl, R.~Colbeck, J.~M. Renes, and R.~Renner, ``The
  uncertainty principle in the presence of quantum memory,'' {\em Nature
  Physics}, vol.~6, pp.~659--662, 2010.

\end{thebibliography}

\begin{thebibliography}{10}

\bibitem{Kent2011}
A.~Kent, W.~J. Munro, and T.~P. Spiller, ``Quantum tagging: Authenticating
  location via quantum information and relativistic signaling constraints,''
  {\em Physical Review A}, vol.~84, p.~012326, 2011.

\bibitem{Buhrman2013}
H.~Buhrman, S.~Fehr, C.~Schaffner, and F.~Speelman, ``The garden-hose model,''
  in {\em Proceedings of the 4th Conference on Innovations in Theoretical
  Computer Science}, ITCS '13, pp.~145--158, ACM, 2013.

\bibitem{MacWilliams1977}
F.~J. MacWilliams and N.~J.~A. Sloane, {\em The Theory of Error-Correcting
  Codes}, vol.~16 of {\em North-Holland Mathematical Library}.
\newblock North-Holland, 1977.

\bibitem{Tomamichel2016}
M.~Tomamichel, {\em Quantum Information Processing with Finite Resources},
  vol.~5 of {\em SpringerBriefs in Mathematical Physics}.
\newblock Springer, 2016.

\bibitem{Watrous2018}
J.~Watrous, {\em The Theory of Quantum Information}.
\newblock Cambridge University Press, 2018.

\bibitem{ledoux1991probability}
M.~Ledoux and M.~Talagrand, {\em Probability in Banach Spaces: Isoperimetry and
  Processes}, vol.~23 of {\em A Series of Modern Surveys in Mathematics
  Series}.
\newblock Springer, 1991.

\bibitem{Bennett1984}
C.~H. Bennett and G.~Brassard, ``Quantum cryptography: Public key distribution
  and coin tossing,'' in {\em Proceedings of the International Conference on
  Computers, Systems and Signal Processing}, vol.~175, pp.~175--179, 1984.

\bibitem{buhrman2014position}
H.~Buhrman, N.~Chandran, S.~Fehr, R.~Gelles, V.~Goyal, R.~Ostrovsky, and
  C.~Schaffner, ``Position-based quantum cryptography: Impossibility and
  constructions,'' {\em SIAM Journal on Computing}, vol.~43, no.~1,
  pp.~150--178, 2014.

\bibitem{RB09}
J.~M. Renes and J.-C. Boileau, ``Conjectured strong complementary information
  tradeoff,'' {\em Physical Review Letters}, vol.~103, p.~020402, 2009.

\bibitem{BCC+10}
M.~Berta, M.~Christandl, R.~Colbeck, J.~M. Renes, and R.~Renner, ``The
  uncertainty principle in the presence of quantum memory,'' {\em Nature
  Physics}, vol.~6, pp.~659--662, 2010.

\bibitem{Winter2016}
A.~Winter, ``Tight uniform continuity bounds for quantum entropies: Conditional
  entropy, relative entropy distance and energy constraints,'' {\em
  Communications in Mathematical Physics}, vol.~347, pp.~291--313, 2016.

\bibitem{Kushilevitz1996}
E.~Kushilevitz and N.~Nisan, {\em Communication Complexity}.
\newblock Cambridge University Press, 1996.

\bibitem{Beigi2011}
S.~Beigi and R.~K\"onig, ``Simplified instantaneous non-local quantum
  computation with applications to position-based cryptography,'' {\em New
  Journal of Physics}, vol.~13, no.~9, p.~093036, 2011.

\bibitem{Junge2021}
M.~Junge, A.~M. Kubicki, C.~Palazuelos, and D.~{P{\'e}rez-Garc{\'i}a},
  ``Geometry of {Banach} spaces: a new route towards position based
  cryptography,'' {\em arXiv-preprint arXiv:2103.16357}, 2021.

\bibitem{Lau2011}
H.-K. Lau and H.-K. Lo, ``Insecurity of position-based quantum-cryptography
  protocols against entanglement attacks,'' {\em Physical Review A}, vol.~83,
  p.~012322, 2011.

\bibitem{tomamichel2013monogamy}
M.~Tomamichel, S.~Fehr, J.~Kaniewski, and S.~Wehner, ``A
  monogamy-of-entanglement game with applications to device-independent quantum
  cryptography,'' {\em New Journal of Physics}, vol.~15, no.~10, p.~103002,
  2013.

\bibitem{Ribeiro2015}
J.~Ribeiro and F.~Grosshans, ``A tight lower bound for the {BB84}-states
  quantum-position-verification protocol,'' {\em arXiv-preprint
  arXiv:1504.07171}, 2015.

\bibitem{Gonzales2019}
A.~Gonzales and E.~Chitambar, ``Bounds on instantaneous nonlocal quantum
  computation,'' {\em IEEE Transactions on Information Theory}, vol.~66, no.~5,
  pp.~2951--2963, 2019.

\end{thebibliography}
\end{document}